\documentclass[11pt,onecolumn]{IEEEtran}
\usepackage{graphicx}
\usepackage{algorithmic}
\usepackage{algorithm} 
\usepackage{amssymb}
\usepackage{subfigure}
\usepackage{enumerate}
\usepackage{amsthm}
\usepackage[top=1.5in, left=0.75in, right=0.75in, bottom=0.75in]{geometry}
\newtheorem{Lem}{Lemma}
\newtheorem{Cor}{Corollary}
\newtheorem{Def}{Definition}
\newtheorem{Thm}{Theorem}

\renewcommand{\baselinestretch}{1.4}

\begin{document}

\abovecaptionskip5pt \belowcaptionskip5pt

\title{How to group wireless nodes together?\\
\large
A survey on Matchings and Nearest Neighbour Graphs}
\author{Anastasios Giovanidis\\
CNRS \& T\'el\'ecom ParisTech - LTCI\\
email: anastasios.giovanidis@telecom-paristech.fr
}
\maketitle
\IEEEpeerreviewmaketitle

\tableofcontents

\newpage

\begin{abstract}
This report presents a survey on how to group together in a static way planar nodes, that may belong to a wireless network (ad hoc or cellular). The aim is to identify appropriate methods that could also be applied for Point Processes. Specifically matching pairs and algorithms are initially discussed. Next, specifically for Point Processes, the Nearest Neighbour and Lilypond models are presented. Properties and results for the two models are stated. Original bounds are given for the value of the so-called generation number, which is related to the size of the nearest neighbour cluster. Finally, a variation of the nearest neighbour grouping is proposed and an original metric is introduced, named here the ancestor number. This is used to facilitate the analysis of the distribution of cluster size. Based on this certain related bounds are derived. The report and the analysis included show clearly the difficulty of working  in point processes with static clusters of size greater than two, when these are defined by proximity criteria.
\end{abstract}

\section{Introduction}

In this overview report, a certain number of questions is identified that has been raised from our previous work \cite{GiovaWCNC12}, \cite{BaccGiovAsilomar13}, \cite{AGFBTWC14}, \cite{GioAlvDeS15} related to cooperation in wireless (cellular) networks. Specifically we investigate the question of what is the best way to group nodes of a wireless network together, based on their location only. Furthermore, we propose possible steps for future research by identifying tools that exist in the literature. 

Apart from the literature survey, the report contains a certain number of novel results on the Nearest Neighbour Graph and the distribution of the size of groups (from now on called clusters). These are found in Section \ref{subGen} and \ref{SecVI}.

\subsection{Dynamic Clusters}

A key problem raised in the modelling and analysis of cooperative wireless networks in \cite{AGFBTWC14}, \cite{BartekPaulFactorial14}, \cite{PapadogiannisDyn08}, \cite{NigamMiHaenngiDec13}, \cite{NigamJour14}, \cite{TanbourgiCOOP14}, \cite{LeeTWC15}, \cite{Sakr14} has been the fact that these clusters were formed in a dynamic way. This is often not desirable, because it requires high flexibility in the inter-cell communication, as well as a large amount of information exchange. Additionally, when the same station takes part in different clusters, it should share its time or frequency resources among them, thus leading to spectral efficiency degradation. Such assumptions on dynamic cluster formation have further created problems in the analysis of interference. As a result in \cite{AGFBTWC14} the far-field approximation had to be made. Furthermore, in the same work, due to the appearance of secondary users, the same Base Station that serves the user, may also create first order interference, in order to serve the secondary users simultaneously. 

Some of these observation were also identified in \cite{AkoumHeathSPA12}, \cite{AkoumHeathJournal13}, where the authors use a clustering method based on random centres distributed themselves as a point process and in  \cite{ParkColor16} where the authors provide an approach towards solution by use of graph coloring.

\subsection{Static Clusters}

To avoid such problems for good, an idea, which we investigate here, is to model cooperative networks with disjoint and fixed clusters. In such case, a base station will only take part in at most one single cluster. Furthermore, planar areas will be permanently assigned to a specific cluster.

Analysis in this way will be facilitated and the results will give more clear evidence of possible performance benefits by clustering Base Stations and sharing their users. What is more, in such settings more complicated Multiple-Input-Multiple-Output (i.e Multiple BSs Multiple Users) cooperation scenarios can be examined, something which was not easy in the previous works, e.g. in \cite{AGFBTWC14} we analysed $2\times 1$ pairs of BSs that serve exactly one user/location. Other authors \cite{AkoumHeathSPA12}, \cite{AkoumHeathJournal13}, \cite{LeeTWC15}, have made different efforts with approximative models. 

\subsection{Static Clusters based on Node Proximity}

We are interested to find grouping methods based on node proximity; hence we do not consider here channel characteristics. The reason is that we search for clusters of fixed size and fixed elements (nodes), whose choice does not depend on the variable parameters of the telecommunication systems, as for example the fading or the user positions. In this sense, our work here aims at \textit{network-defined and fixed} clusters, to differentiate from a user-driven selection, done in our previous work \cite{BaccGiovAsilomar13}, \cite{AGFBarXiv13} as well as the related research by Thiele \textit{et al.} in \cite{ThieleAsilomar12}.

The criterion of relative distances for defining clusters is reasonable for the following reasons. It is related to the path-loss factor of the channel power. Clusters should in general be designed to have a stable structure for long time-intervals, so that they can provide stable link quality. The quality of a link is best described by the path-loss if we average over the fast-fading, which varies over time. The path-loss is a decreasing function of the distance from the transmitter and it deteriorates due to the path-loss exponent. Taking this into account, small relative distances will allow BSs to exchange messages more reliably (even allow the installation of fixed optical fiber communications between them, since clusters do not change). They will also allow a user to be served more efficiently, by a larger group of BSs, whose distance from the user is small, and hence their transmitted signal will on average be received in a satisfactory power level at the user's receiver. Consequently, if a group of atoms lies close enough to each other and relatively far away from other atoms, their potential cluster serving users in cooperation will be beneficial for the network performance. On the one hand, the BSs of the cluster will transmit in cooperation without causing interference to each other and at the same time they may increase their beneficial signal through network MIMO techniques. On the other hand, these BSs will be far away from other atoms, hence the inter-cluster interference that they cause will be low. Cluster choice based on path-loss has already been studied in a previous work of one of the authors \cite{GiovaWCNC12}. In the latter work however, the user positions were assumed known and influenced the cluster choice, since the criterion was the long-term service satisfaction above a predefined threshold.

\subsection{Open Questions for Static Clustering}

We identify a set of open problems, that can be summarised in the following questions:

\begin{enumerate}
\item \textit{How do we define the static clusters given a realisation of a Point Process? In other words, which atom will cooperate with which neighbour to form clusters? What is the grouping law?}
\item \textit{How large should these clusters be?}
\item \textit{Once the clusters are defined, how do we seperate and assign the planar areas to each cluster?}
\end{enumerate}

After having answered such questions, further performance analysis is possible.

\vspace{+1cm}
This report is organised as follows:

\begin{itemize}
\item In Section \ref{SecI}, some existing results from the literature on the problem of \textit{matching pairs }are presented. These works treat finite sets of atoms and the section includes also an example of how a stable matching (pairing) could actually work. It concludes with a presentation of the problem of matching pairs for a planar Point Process in \ref{subSecII}. The main question is whether the existing results can be extended to the infinite case. After that, the next relevant question, is whether we can redefine the Voronoi tessellation of the plane taking as centres not anymore the individual atoms but rather the pairs, or the linear segments that connect them.

\item Section \ref{SecIII} considers the important topic of advancing beyond the concept of cooperation in pairs and suggests ways of cooperation in pairs, triplets etc... taking into consideration the concepts already discussed in the previous section.

\item Section \ref{SecIV} approaches the problem of static clustering specifically for nodes produced by a realisation of a Point Process. The section investigates ways to group nodes only based on proximity. Two methods are presented from the literature, namely the Nearest Neighbour and the Lilypond model. Common properties of the two models are presented and their differences are explained. 

\item Section \ref{SecV} focuses on the Nearest Neighbour graph and provides main results from the literature, on the structure and size of such clusters. A novel theorem is included, specifically, Theorem 3, where bounds for the size of a branch of the cluster from a Nearest Neighbour graph for Poisson Point Processes are derived. These are based on the so called, \textit{generation number}, already introduced in the literature.

\item Section \ref{SecVI} introduces another metric for the size of a cluster, formed by a variation of the nearest neighbour model. The metric is named here the \textit{ancestor number}. This number does not count only the size of one branch of the cluster, but rather aims at approximating its total size. Based on geometric bounds and for Poisson Point Processes, the ancestor number is bounded above and below in Theorem 4. The section concludes with bounds on the tail probability of the size of a cluster produced by this variation of the nearest neighbour graph.

\end{itemize}

\section{Matching Problems}
\label{SecI}

The literature is very rich especially in "pairing" problems (marriage, roommates) for finite sets. These relevant problems are discussed in the following subsections.

%\begin{itemize}
%\item 
\subsection{The Stable Marriage Problem}
In this problem first presented and investigated by Gale and Shapley (1962) \cite{GaleShapley62}, there are exactly $n$ men and $n$ women who wish to be paired in a stable way. Each person (woman or man) ranks those of the opposite sex in accordance to her/his preferences for a marriage partner. 

Suppose we suggest a configuration of $n$ marriages between these men and women. The configuration is \textit{unstable} if under it there are a man and a woman, not married to each other, but who prefer each other (according to their ranking) to their actual mates. We call this a \textit{blocking pair} for the configuration.

An important conclusion of the paper is that "\textit{there always exists a stable set of marriages}". In other words, we can always find a configuration that is stable, for any preference list of the users. The proof also provides an algorithm to find such a stable configuration, which may not be unique.
The algorithm is based on a sequence of \textit{proposals} and \textit{considerations}, meaning that a girl (say) may consider a boy who proposed in some round, but can reject him in a next round when a better proposal may come. In this way, the algorithm converges to a configuration without blocking pairs.

%\item 
\subsection{The Stable Roommates Problem}
This problem is an extension of the Stable Marriage Problem, and was discussed briefly as a special case, in the original Gale-Shapley paper. It was thoroughly treated later on, in the work of Irving (1985) \cite{IrvingRoomMate}. The difference with the above, is that there exists a single set of $n$ persons who wish to be paired  with each other, with the aim e.g. to choose double rooms in a dormitory. Each person, will rank the remaining $n-1$ in a way that lower ranking corresponds to higher preference.

The question is again, as in the stable marriage case, whether there always exists a stable configuration of roommates. This problem, which does not differentiate between men and women is of better relation to our work, where we meet a similar problem of pairing atoms in a realization of a point process.

The algorithm suggested by Irving has two phases. The first is a sequence of suggestions and considerations, where each person rejects any poorer proposals than the one that she/he already has at hand. The result of the first phase is a reduction of the preference list of the candidates. The second phase involves the identification from this list of certain cycles of users with specific preference characteristics. These cycles will determine which persons from the reduced preference list can further be removed. The result of the second phase is either a stable matching (pairing) or the conclusion that no such matching exists for the current problem. Note that the negative conclusion can potentially come already after the first phase. It is the case when some person is left without a proposal to consider after termination of the sequence.

The basic difference between this problem and the problem of Stable Marriage is that \textit{the former may not always obtain a stable solution}.

%\item 
\subsection{Acyclic Networks}

The authors of the reference works  \cite{MathieuMatchP2P}, \cite{GaiP2Peuro} have identified a very important feature of acyclic preference lists. An acyclic list is one where there exists \textit{no preference cycle} between $k\geq 3$ persons $p_1,p_2,\ldots,p_k$, where $p_i$ prefers $p_{i+1}$ to $p_{i-1}$ (modulo $k$).

Specifically, they proved that \textit{an acyclic matching preference instance (i.e. list) always has a unique stable configuration}. As a consequence \textit{any sufficiently long sequence of active initiatives leads to the unique stable configuration}.

An \textit{initiative} here is a probing between two pairs of some possibly suboptimal configuration, which - when \textit{active} - results in a blocking pair (i.e. breaking of the two pairs and connection of one element of the one with some other from the second pair, which is mutually preferable).

The above results easily extend to so-called \textit{b-matchings}, where one person has a quota $b\left(p\right)\geq 1$ on the number of mates which he can be related with.

The importance of the result lies on the fact that, most preference lists exhibit the acyclic property. Specifically, if we use a mark $m\left(i,j\right)$ to denote the preference of $i$ over $j$, we can identify three preference categories, which are always acyclic:

\begin{itemize}
\item \textit{Global preferences}, where $m\left(i,j\right)=m\left(j\right)$ and the mark is fixed for each person (e.g. bandwidth, capacity etc) and is neighbor-independent.
\item \textit{Symmetric preferences}, where $m\left(i,j\right)=m\left(j,i\right)$ and the marks as seen as equal from both sides of the edge of the graph. Such case is found very often in real world network problems, as when the distance, or latency is used as a mark.
\item \textit{Complementary preferences}.
\end{itemize}

The key property of the acyclic networks is that \textit{any non-trivial acyclic preference instance always has at least one loving pair}.

A loving pair $\left\{p,q\right\}$ is a pair between persons, such that person $p$ is ranked first by person $q$ and vice versa. Hence, loving pairs are unbreakable. Finding them, also provides an algorithm of finding the unique stable matching which converges in a finite number of steps. 

Specifically, \textit{for any acyclic preferences instance, starting from any initial configuration C, there exists a sequence of at most $\frac{B}{2}$ initiatives leading to the stable solution, where $B=\sum_{p\in P}b\left(p\right)$.}

The algorithm is based simply in probing (making an initiative between) different pairs at a time, to identify blocking pairs, which will be more stable. Once a loving pair is found, it is removed from the list and the process is repeated for the remaining persons.  
%\end{itemize}

\subsection{Example of Matching pairs in an Acyclic Network}

We consider in this subsection a concrete example of an acyclic network, where we apply the results of \cite{MathieuMatchP2P} and the algorithm which converges to a stable matching in pairs. Specifically, we consider an instance of eight atoms located on the two-dimensional plane as shown in Fig. \ref{points8exemp}.

\begin{figure}[ht]
\centering
 \includegraphics[width=12cm]{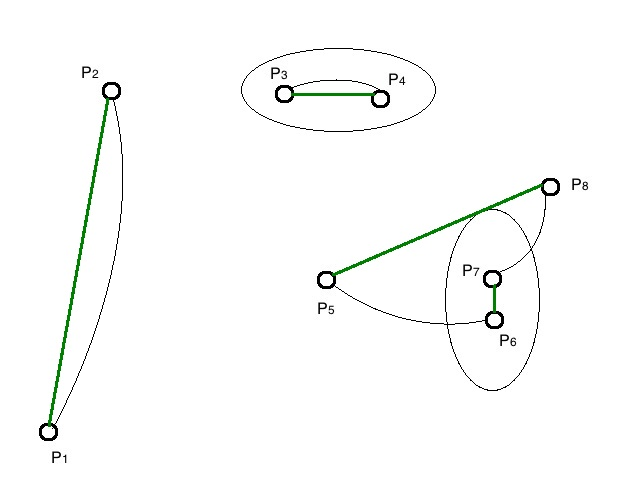}
	\caption{Example of 8 atoms scattered on the two-dimensional plane. The black curved lines are the edges of the initial matching and the green straight lines the final optimal ones. The two pairs within the ellipses are \textit{loving pairs}.}
\label{points8exemp}
\end{figure}
We can derive from the above figure the preference list of the atoms, based on their distance. Observe that when distance is used as mark, the resulting network is symmetric and consequently acyclic. Each atom ranks the rest seven atoms from 1 to 7, with 1 being its closest neighbor and 7 its farthest.

\begin{figure}[h]
\centering
\begin{tabular}{| c | c | c | c | c | c | c | c | c |}
\hline
 & $P_1$ & $P_2$ & $P_3$ & $P_4$ & $P_5$ & $P_6$ & $P_7$ & $P_8$\\\hline
 $P_1$ & X & 4 & 7 & 7 & 7 & 7 & 7 & 7 \\\hline
 $P_2$ & 2 & X & 2 & 5 & 6 & 6 & 6 & 6 \\\hline
 $P_3$ & 3 & 1 & X & 1 & 5 & 5 & 5 & 5 \\\hline
 $P_4$ & 4 & 2 & 1 & X & 3 & 4 & 4 & 3 \\\hline
 $P_5$ & 1 & 3 & 3 & 2 & X & 3 & 3 & 4 \\\hline
 $P_6$ & 5 & 6 & 6 & 6 & 2 & X & 1 & 2 \\\hline
 $P_7$ & 6 & 5 & 5 & 3 & 1 & 1 & X & 1 \\\hline
 $P_8$ & 7 & 7 & 4 & 4 & 4 & 2 & 2 & X \\\hline
\end{tabular}
\label{Pref8exemp}
\caption{Preference list for the 8 atoms of the above example, based on the distance from each other.}
\end{figure}

A first important observation from the preference list is that the pairs $\left\{3,4\right\}$ and $\left\{6,7\right\}$ are loving pairs, i.e. they consider each other first closest neighbours

\begin{eqnarray}
P_3 \stackrel{1}{\leftrightarrow} P_4   & \& &  P_6 \stackrel{1}{\leftrightarrow} P_7. \nonumber
\end{eqnarray}

We initialize a constellation $\mathcal{C}_0$ by pairing at random. The initial pairing is shown on Fig. \ref{points8exemp} with the black curved lines between atoms. The ellipses surrounding the two pairs highlight the loving pairs in our example. We can write down the initial constellation $\mathcal{C}_0$

\begin{eqnarray}
P_1 \leftrightarrow P_2 & & \left(P_1 \stackrel{2}{\rightarrow} P_2, \ \ P_2 \stackrel{4}{\rightarrow} P_1\right)\nonumber\\
P_3 \leftrightarrow P_4 & & \left(P_3 \stackrel{1}{\rightarrow} P_4, \ \ P_4 \stackrel{1}{\rightarrow} P_3\right)\nonumber\\
P_7 \leftrightarrow P_8 & & \left(P_7 \stackrel{2}{\rightarrow} P_8, \ \ P_8 \stackrel{1}{\rightarrow} P_7\right)\nonumber\\
P_5 \leftrightarrow P_6 & & \left(P_5 \stackrel{2}{\rightarrow} P_6, \ \ P_6 \stackrel{3}{\rightarrow} P_5\right)\nonumber.
\end{eqnarray}

The notation $P_i \stackrel{n}{\rightarrow} P_j$ means that the atom $P_j$ is ranked $n$-th by $P_i$.  The steps of the algorithm are as follows:

\begin{enumerate}
\item The pair $\left\{P_3,P_4\right\}$ is a loving pair and hence can be removed from the list. 
\item From the two pairs  $\left\{P_7,P_8\right\}$ and  $\left\{P_5,P_6\right\}$ a blocking pair can be found $\left\{P_6,P_7\right\}$, which is also a loving pair. Hence the above two pairs break, the loving pair  $\left\{P_6,P_7\right\}$ is formed and removed from the preference list. Furthermore, the pair  $\left\{P_5,P_8\right\}$ is formed.
\item The remaining pairs are $\left\{P_1,P_2\right\}$ and $\left\{P_5,P_8\right\}$. 
\item Observing the list of preferences, we see that no blocking pair can be created from the two, hence the two remaining pairs are stable.
\item The final matching is shown in Fig. \ref{points8exemp} with green straight lines.
\end{enumerate}

\subsection{Matchings on infinite sets and Point Processes}
\label{subSecII}
1) \textit{Infinite Sets and Matching - Basic Questions}

The great \textit{challenge} is to find out whether the results of existence of a stable matching for finite sets can be extended to the infinite case, where the atoms of a Point Process form preference lists - possibly having as criterion the distance of each one with all its neighbours. Important questions to be answered here are the following:

\begin{itemize}
\item Does such a stable matching \textit{always exist} for $n\rightarrow \infty$ and for each type of Point Process (Poisson, Mattern, etc.))?
\item Is there a \textit{unique matching}?
\item Is there a \textit{polynomial-time} algorithm which is guaranteed to converge to a possible stable matching? Can we propose a \textit{decentralized} solution?
\end{itemize}

In the results of the previous paragraph, we saw that for finite acyclic networks, such as those which use the distance as mark (i.e. relevant to the Point Processes and the problems of interest), there exists a \textit{unique stable configuration}. Hence, the conjecture that a unique stable matching of pairs for $n\rightarrow \infty$ exists, when the criterion is the distance, seems very logical (still needs to be proved).

Problems arise, only when there exists a preference cycle between $k\geq 3$ atoms $p_1\ldots p_k$, where $p_i$ prefers $p_{i+1}$ to $p_{i-1}$. But this is impossible due to the distance ranking.

The special case, where $k\rightarrow\infty$ should not be excluded. This is the case where each atom $p_i$ is the closest neighbour for $p_{i-1}$, but its own closest neighbour $p_{i+1}\neq p_{i-1}$ and lies outside the ball with center $p_{i-1}$ and having $p_i$ on its boundary (repeated infinitely many times). Such cases of atom placement are relevant with the so called \textit{lilypond model}, first appearing in \cite{HagMeest96} for the Poisson Point Process and further studied by Last and Penrose (2010) \cite{LastLilypond10}. 

Specifically, it is proven in \cite[Lemmas 2.1, 2.2]{LastLilypond10} under certain assumptions, that no such infinite preferential cycle of atoms can exist. 

\vspace{+1cm}

2) \textit{Further Questions related to Geometry}

Relevant further questions -  after the above have been answered - are related to the \textit{distribution of geometric characteristics} of the matchings on the two-dimensional plane.

\begin{itemize}
\item What can we say about the distribution of the distance between two matched partners? 
\item Can we define a new \textit{tessellation of the plane} - like the 1-Voronoi tessellation - which is related to the line segments of matched pairs?  
\item How does the distribution of neighbours around an atom affect the results on matching? (mutual neighbour distribution)
\end{itemize}

\begin{figure}[h]
\centering
 \includegraphics[trim = 5mm 135mm 20mm 10mm, clip, width=12cm]{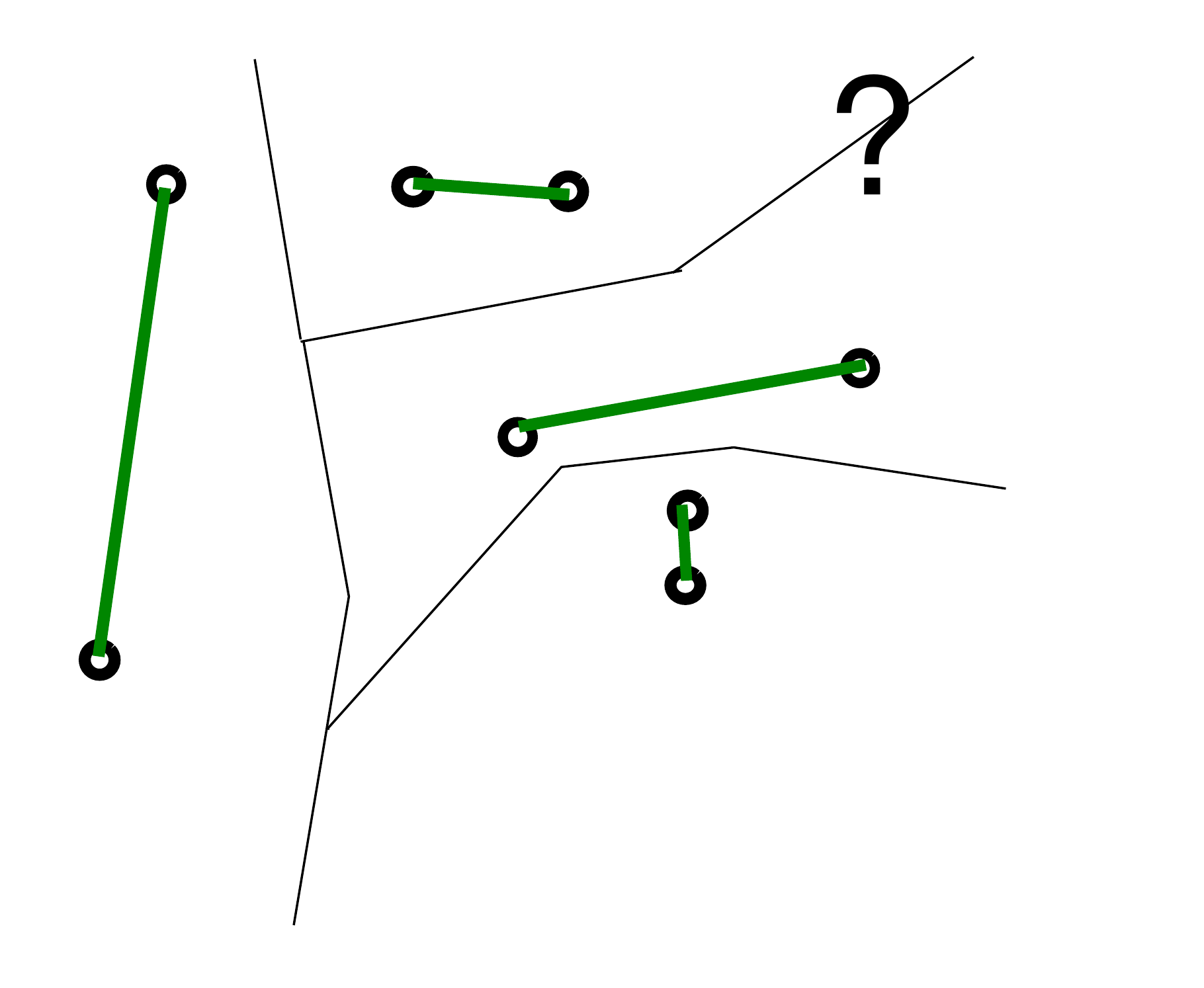}
	\caption{Example of how the tessellation of the 2-dimensional plane could look like, taking as centers the linear intervals defined by the previous matching. We show again the example of the eight atoms similar to the previous section.}
\label{Voronoi8exemp}
\end{figure}

\vspace{+1cm}
3) \textit{Voronoi Tessellation of Segments}

Especially the second question in the previous subsection is of particular interest. After determining the appropriate matchings, we would like to separate the plane into \textit{compact} and maybe \textit{convex} subregions, each one of which should be related to exactly one cooperation cluster. 
An example is shown in Fig. \ref{Voronoi8exemp}.

By connecting the atoms of each cluster, line segments and polygons are formed, and the question on the appropriate tessellation is now stated more generally. 

As shown in \cite{VoronoiMedialAxis02}, the Voronoi cell definition is not restricted to district points, when referring to sites. So, the concept of the Voronoi cell can be generalized to sets of points like e.g. the matched atom pairs, or the linear intervals. Furthermore, an interesting concept relevant to the Voronoi tessellation for segments, is the so-called \textit{medial axis}, and the relation between these two notions.

Another relevant publication, which extends (or rather provides an alternative notion to) the Voronoi tessellation, is the work by Hoffman, Holroyd and Peres (2006), which relates planar points to cells of atoms in a more "fair" way, so that all subregions have equal volume \cite{PeresPoissonLebesgue06}. The volumes are defined by a so-called appetite, modeled by the parameter $\alpha\in\left(0,\infty\right]$ and the points are assigned to atoms (and relevant Voronoi cells) based on the Gale-Shapley stable marriage principle. When $\alpha\rightarrow \infty$ the tessellation is identical to the Voronoi tessellation of the plane.

\section{Clusters beyond just Pairs}
\label{SecIII}

It is reasonable to continue our research beyond the constraint that only pairs of Base Stations are allowed as clusters and all atoms should belong to exactly one such pair. This is a \textit{very restrictive assumption that does not necessarily model the real systems and how cooperative clusters might work within them}.

A different approach would be to allow pairs, triplets, quadruplets, etc. of atoms cooperate, while at the same time certain atoms might not cooperate at all. This approach sounds more reasonable, when atoms form cooperative entities with increasing size depending on their distance and only when this is necessary and helpful. In other words, it might not pay much to form a cooperative pair between two atoms that lie very far away from each other, just for the sake of forming a pair.

A possible way to construct such clusters will be as follows:

\begin{enumerate}
\item Start with the realization of a Point Process. The current constellation is the \textit{non-cooperative case}.

\item Identify pairs of atoms that are \textit{mutual first neighbours} to one another as in \cite{GioAlvDeS15}. 
\begin{eqnarray}
P_{i}\stackrel{1}{\leftrightarrow}P_j.\nonumber
\end{eqnarray}

These neighbours will form the pair-clusters and the new constellation will be called \textit{2-cooperative case}. Of course all the atoms that do not find a pair with mutual first neighbour will be treated as single atoms, so that pairs and single Base Stations coexist.

\item Among the cooperative pairs, identify triplets of atoms, such that the following chain of preference (related to distance) holds 
\begin{eqnarray}
P_{i-1}\stackrel{1}{\rightarrow}P_i\stackrel{1}{\leftrightarrow}P_{i+1} & or & P_{i-1}\stackrel{1}{\leftrightarrow}P_i\stackrel{1}{\leftarrow}P_{i+1}.\nonumber
\end{eqnarray}
These obviously constitute the \textit{3-cooperative case}. In such scenario single atoms, pairs and triplets coexist. In the formation of such triplets, possibly a constraint on the distance should be considered, so that cases where a third single atom which is far away is included in the triplet are avoided. 

(Notice: Remember that the network is acyclic so we do not consider the third case where $P_{i+1}$ may have first preference to $P_{i-1}$ or the other way round. Also, the case of three equidistant atoms is impossible and has zero probability in the Poisson Point Process case.)

\item Among the cooperative triplets identify chains of preference with four atoms, and so on...
\item ...
\end{enumerate}

This type of clustering is based and builds on the previous definition of \textit{loving pairs} in the subsection for acyclic networks. The clusters are formed as cooperative entities including atoms which have only relations of first preference (first closest neighbour) at least in one direction. An example of such a clustering is shown in the following Fig. \ref{BeyondPairexemp}.

\begin{figure}[h]
\centering
 \includegraphics[trim = 5mm 185mm 20mm 0mm, clip, width=12cm]{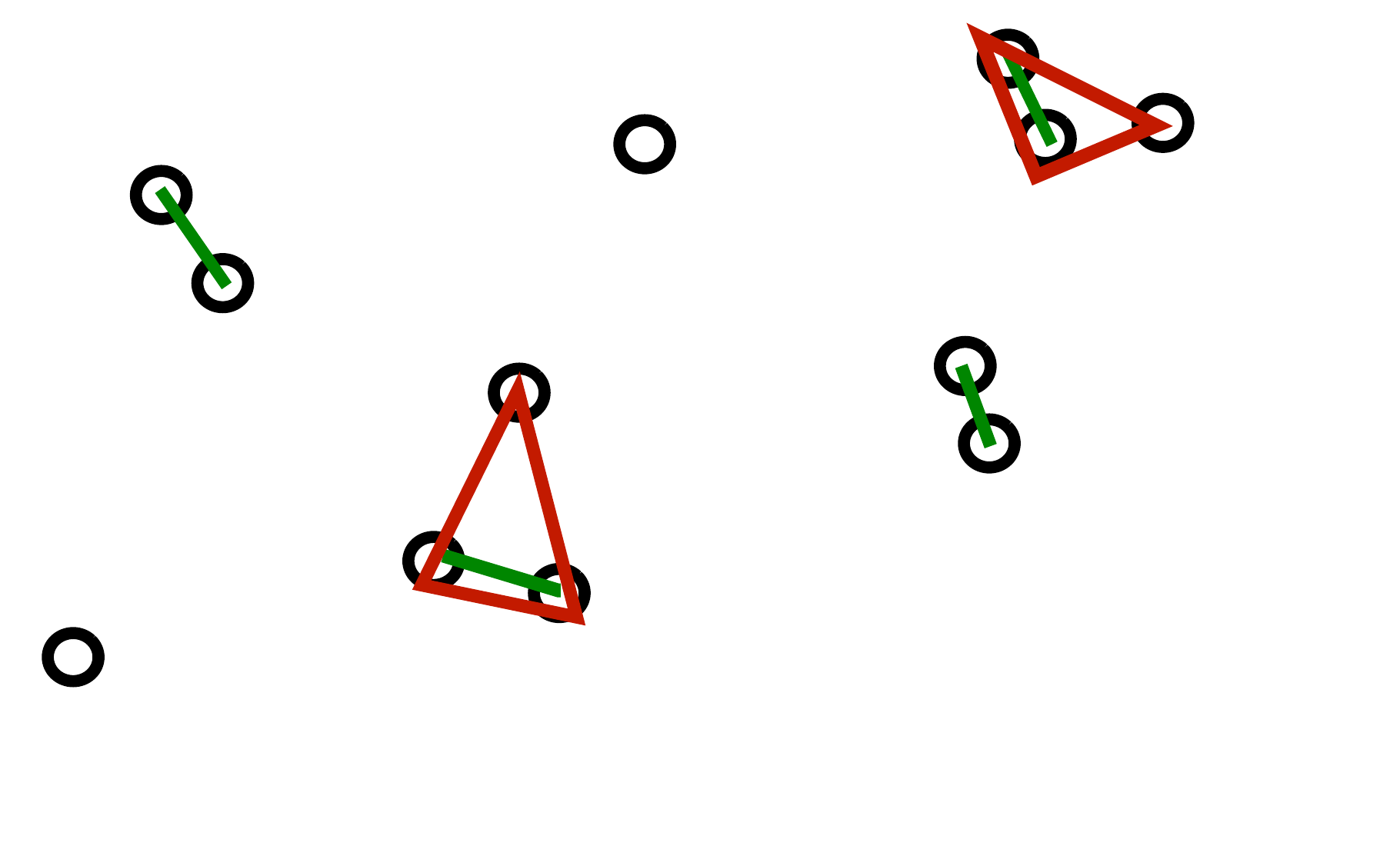}
	\caption{Example of the proposed matching with no-cooperation, 2-cooperation and 3-cooperation case. The pairs are shown with green linear intervals and the triplets with red.}
\label{BeyondPairexemp}
\end{figure}

After such clusters are formed consideration of Voronoi tessellations follow. Finally each case of 2-, 3-,$\ldots$ cooperation is compared to the case of no cooperation in order to derive possible benefits related to performance measures such as coverage and throughput.

Other types of cluster formation are also under discussion.

The direct benefit of such an approach is however that the clusters are isolated entities and each one of these has a defined planar area of interest/association. There is no overlap either of planar Voronoi cells or of atoms that might belong to the same cluster. This helps a lot in the analysis and models a way of function that has been suggested as appropriate for cooperative communications systems.

\section{Clusters for Point Processes}
\label{SecIV}

Let us consider a homogeneous point process (p.p) $\Phi$ in $\mathbb{R}^2$ with non-negative density $\lambda>0$. When referring to a Poisson Point Process we will use the abbreviation p.p.p. One realization of the process is $\phi$ and can be described by the infinite set of atoms $\left\{\mathbf{z}_i\right\}$, where $\mathbf{z}_i=\left(x_i,y_i\right)$. Each realization shows a possible deployment of single antenna Base Stations (BSs) on the plane. We wish to group these BSs (or atoms) into \textit{disjoint cooperative clusters}, with possibly different sizes, where size here means the cardinality of atoms included in each cluster. To be more formal, a cluster is defined to be a finite subset $\mathcal{C}\left(\phi\right)$ of the realization $\phi$ (we omit $\phi$ from now on unless we refer to different realizations of the p.p.), with cardinality $card\left(\mathcal{C}\right)$. Exactly as the atoms of $\phi$ are somehow enumerated and indexed by $i$, we enumerate and index the clusters using the index $m$. We consider clusters of atoms with the following two properties: (a) the set of all clusters constitutes a partition of $\phi$, hence their union exhausts the infinite set of atoms and  (b) they are disjoint subsets of $\phi$, meaning that the intersection between any pair of distinct clusters is empty

\begin{eqnarray}
\label{ClusterProp1}
\bigcup_{m=1}^{\infty} \mathcal{C}_m & = & \phi,\\
\label{ClusterProp2}
\mathcal{C}_m\cap\mathcal{C}_n & = & \emptyset,\ \ \forall m\neq n.
\end{eqnarray}

For the partition of $\phi$ into clusters, we intend to use rules that depend only on the geometry. In other words, given a realization of the p.p., an atom $\mathbf{z}_i$  will take part in a cluster, based only on its relative distance to the rest of the atoms $\phi\setminus\left\{\mathbf{z}_i\right\}$ as well as their own relative position. 

We consider in this report different ways to connect the atoms of a p.p. and define the cooperative clusters, based solely on $\phi$. A connection between two atoms means that these cooperate and belong to the same cluster. These two models have been presented and studied partly in the work of H\"aggstr\"om and Meester \cite{HagMeest96}. Specifically we will use:

\begin{enumerate}
\item the \textbf{nearest neighbour (NN) model}. Given the realization $\phi$ we connect each $\mathbf{z}_i$ to its \textit{nearest neighbour} by an undirected edge. This results in a graph $\mathcal{G}_{NN}$, which is well defined at least for a p.p.p. where no two inter-atom distances are the same a.s. and hence each atom has a unique first neighbour. However, an atom can be a nearest neighbour for a set of atoms (possibly empty).
\item the \textbf{lilypond (LL) model}. We construct the graph $\mathcal{G}_{LL}$ dynamically as follows. Starting by the realization $\phi$, we assume at time $t = 0$ that there is a ball of radius $t$ (hence $0$) centered on each atom $\mathbf{z}_i$. Then we let time $t \uparrow$ evolve and the radii of these balls grow (linearly in $t$ and all with the same speed). As soon as a ball hits another ball, it stops growing forever. Notice that the other ball could be itself in a phase of growth (or not). The time instant that the hitting takes place, say $t_n =t$, is saved and gives the ball's radius $r_i := t_n$. In this way, a sequence of ball-touching times (equivalently of radii) is formed $t_1, t_2,\ldots,t_n,\ldots$ and we are interested in the limiting configuration as $t\rightarrow\infty$. Each time instant a ball touches another one and stops growing, an undirected link is drawn at the graph $\mathcal{G}_{LL}$. The lilypond model falls in the category of the hard sphere Boolean model, where balls of different radii centered at atoms of the p.p.p. do not overlap but are possibly tangential.
\end{enumerate}

In general the two models result in different partitions of $\phi$ and different sets of clusters. There are however certain interesting \textit{properties} (\textbf{P.x}) shared by both models:

\begin{enumerate}[I)]
\item The cluster formation is independent of the density $\lambda$ of the p.p.p..
\item The graph $\mathcal{G}$ (either $\mathcal{G}_{NN}$ or $\mathcal{G}_{LL}$) is disconnected, i.e. there exist two atoms that are not connected by any path. 
\item Each resulting cluster $\mathcal{C}$ does not contain \textit{cycles}, it is a \textit{tree} and hence the graph $\mathcal{G}$ is a \textit{forest}.
\item The graph $\mathcal{G}$ contains a.s. no infinite component, i.e. it does not percolate \cite[Th.2.1 and Th.5.2]{HagMeest96}. Consequently, the cardinality of each cluster $card\left(\mathcal{C}\right)$ is a.s. finite.
\item There exists no isolated node in the graph $\mathcal{G}$, i.e. there exists no cluster of $card\left(\mathcal{C}\right)=1$. For the case $\mathcal{G}_{NN}$ this is because all atoms have a nearest neighbour. For the case $\mathcal{G}_{LL}$ because each ball will eventually touch another one as $t\rightarrow\infty$ (the probability of an empty ball for p.p.p. is $e^{-\lambda\pi t^2}\stackrel{t\uparrow}{\rightarrow}0$).
\end{enumerate}

We first show in Fig.\ref{fig:COOPclustNN}-\subref{fig:COOPclustLL}, the $\mathcal{G}_{NN}$ and $\mathcal{G}_{LL}$ graphs produced when applying the nearest-neighbour and lilypond rules on an example realization of a p.p.p. with density $\lambda=2$ [atoms/$m^2$]. Furthermore, we better explain how the $\mathcal{G}_{LL}$ graph is formed as the balls of radius $t$ increase and touch each other. For this, we illustrate four instants of their growth (including the end instant) in Fig.\ref{fig:LilypondEVE20}-\subref{fig:LilypondEVE83}. 

Although the two graphs share many common properties as presented above, it is obvious from the subfigures Fig.\ref{fig:COOPclustNN}-\subref{fig:COOPclustLL} that they actually are different. Of course they have similarities, however there clearly are cases where different clusters are formed. To give an example, let us focus on the largest cluster at the right handside of the plane denoted here by $\mathcal{C}_{o}$, which in the NN graph has cardinality $card\left(\mathcal{C}^{(NN)}_{o}\right)=7$ and in the graph LL has cardinality $card\left(\mathcal{C}^{(LL)}_{o}\right)=5$. A doublet in the lilypond model appears as an independent cluster of cardinality 2, although  in the NN graph it belongs to the large cluster. This occurs because in the LL model an edge of the graph may connect atoms that are not necessarily nearest neighbours. This is illustrated for clarity in Fig.\ref{NNvsLLexemp}. For the atom $\left(2,0\right)$ its nearest neighbour is actually $\left(2,1.5\right)$ (inside the dashed circles). However, the latter is linked with its own nearest neighbour  $\left(2.6,1.5\right)$, while due to the growth process of the LL model, the balls centered at $\left(2,0\right)$ and $\left(2,1.5\right)$ do not touch each other, due to the presence of the atom at the origin $\left(0,0\right)$.

In the paragraphs below, we will analyze the two models separately. Furthermore, we will provide results related to their size and probability, together with a comparison between them.

\begin{figure}[ht]    
\centering  
\label{fig:COOP1}
	 		\subfigure[Clusters in the NN model.]{          
           \includegraphics[trim = 40mm 75mm 40mm 70mm, clip, width=0.46\textwidth]{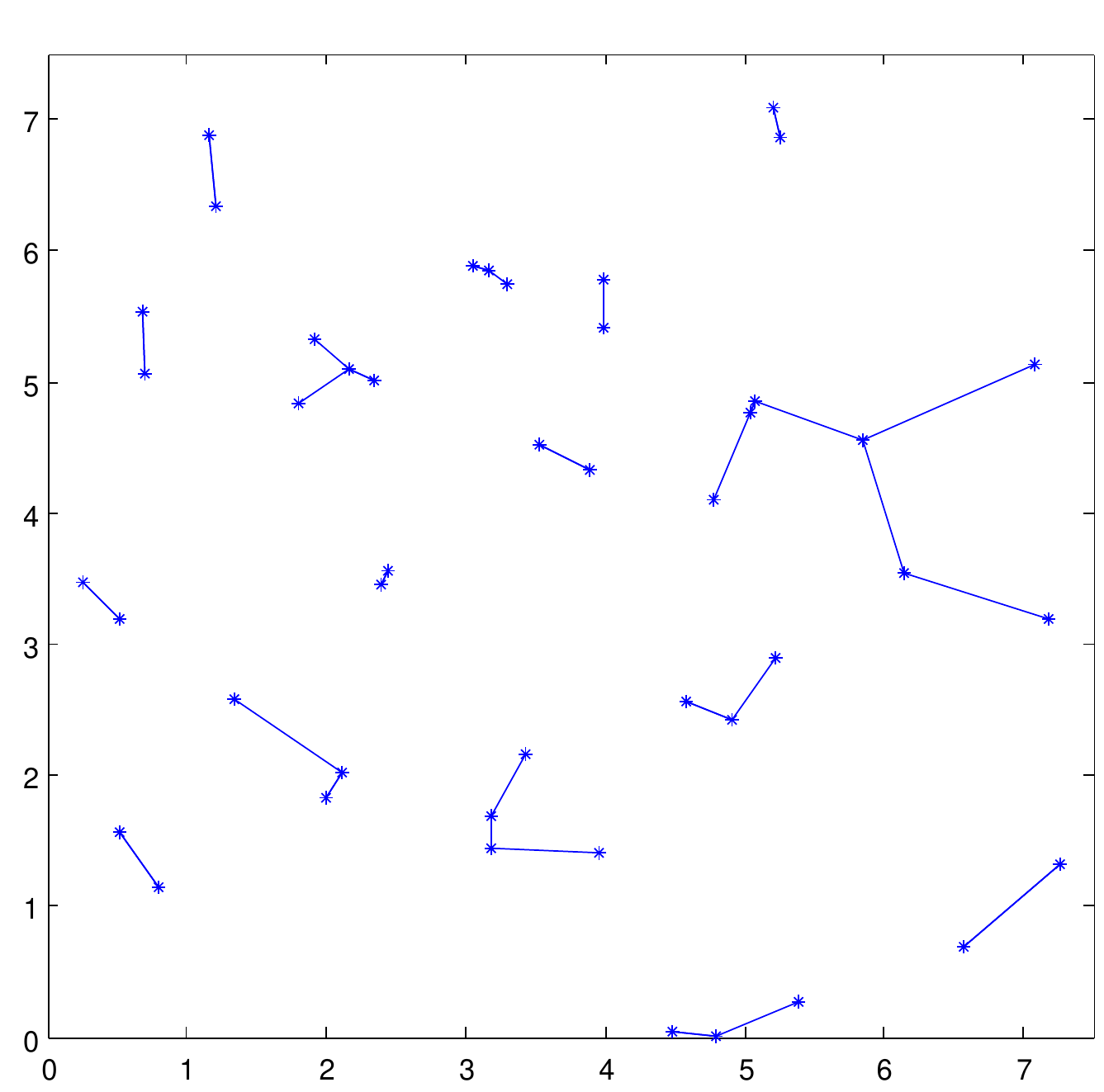}
           \label{fig:COOPclustNN}
           }
            \subfigure[Clusters in the LL model.]{          
           \includegraphics[trim = 40mm 75mm 40mm 70mm, clip, width=0.46\textwidth]{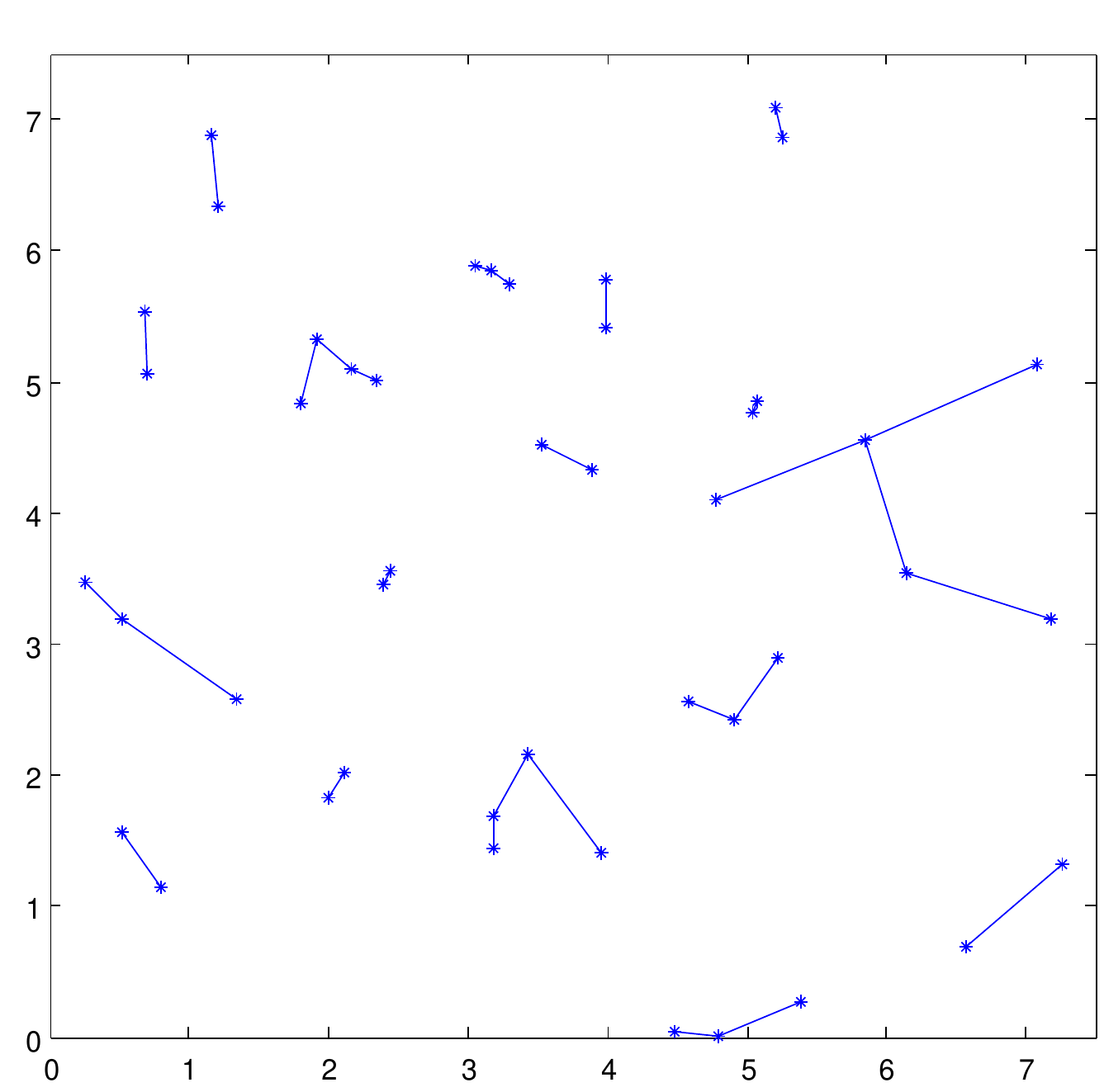}
           \label{fig:COOPclustLL}
           }
           \caption{Example of cooperation clusters for the same example node topology. (a) Nearest-neighbour (NN) (b) Lilypond (LL).}
\end{figure}

\begin{figure}[ht!]    
\centering  
\label{fig:LilypondEVE}
	 		\subfigure[Evolution of the Lilypond $t=0.20$.]{          
           \includegraphics[trim = 30mm 60mm 25mm 60mm, clip, width=0.46\textwidth]{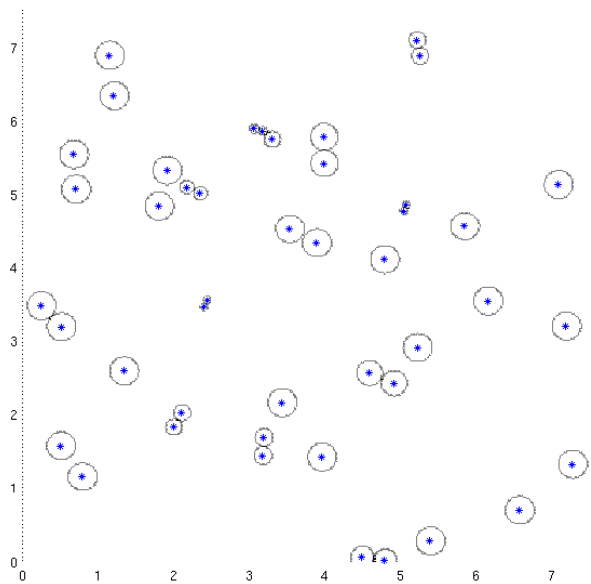}
           \label{fig:LilypondEVE20}
           }
            \subfigure[Evolution of the Lilypond $t=0.40$.]{          
           \includegraphics[trim = 30mm 60mm 25mm 60mm, clip, width=0.46\textwidth]{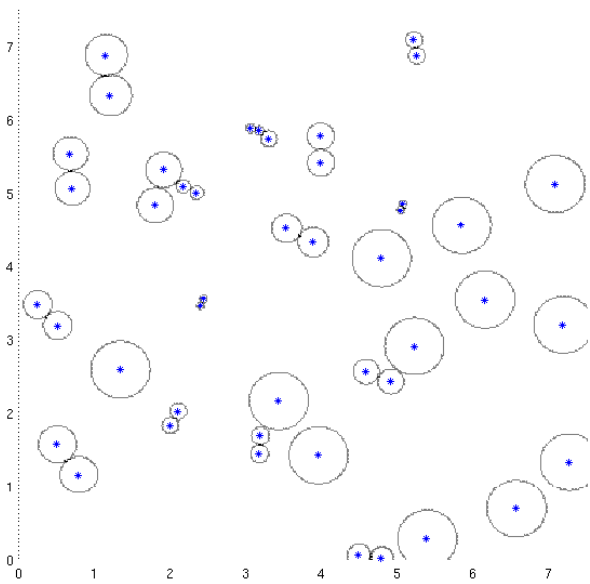}
           \label{fig:LilypondEVE40}
           }
           \subfigure[Evolution of the Lilypond $t=0.60$.]{          
           \includegraphics[trim = 30mm 60mm 25mm 60mm, clip, width=0.46\textwidth]{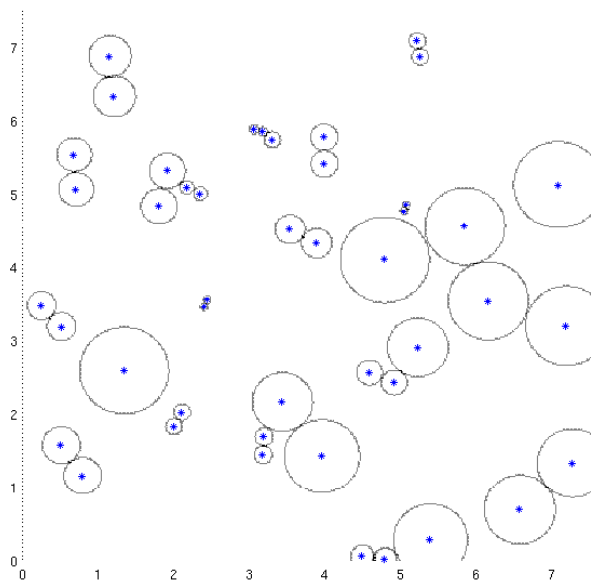}
           \label{fig:LilypondEVE60}
           }
            \subfigure[Evolution of the Lilypond $t=0.83$.]{          
           \includegraphics[trim = 30mm 60mm 25mm 60mm, clip, width=0.46\textwidth]{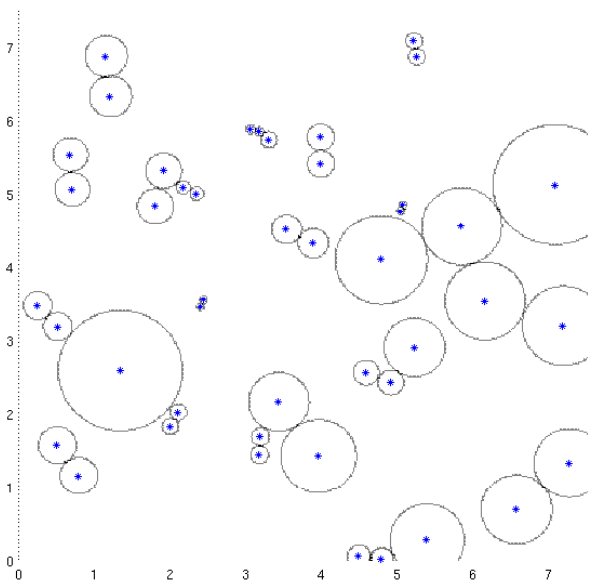}
           \label{fig:LilypondEVE83}
           }
           \caption{Four instants of growth for the Lilypond model using the same example topology as in previous Figure. (a) $T=0.20$, (b) $T=0.40$, (c) $T=0.60$, (d) end of growth $T=0.83$.}
\end{figure}

\begin{figure}[h]
\centering
 \includegraphics[trim = 15mm 65mm 10mm 75mm, clip, width=0.6\textwidth]{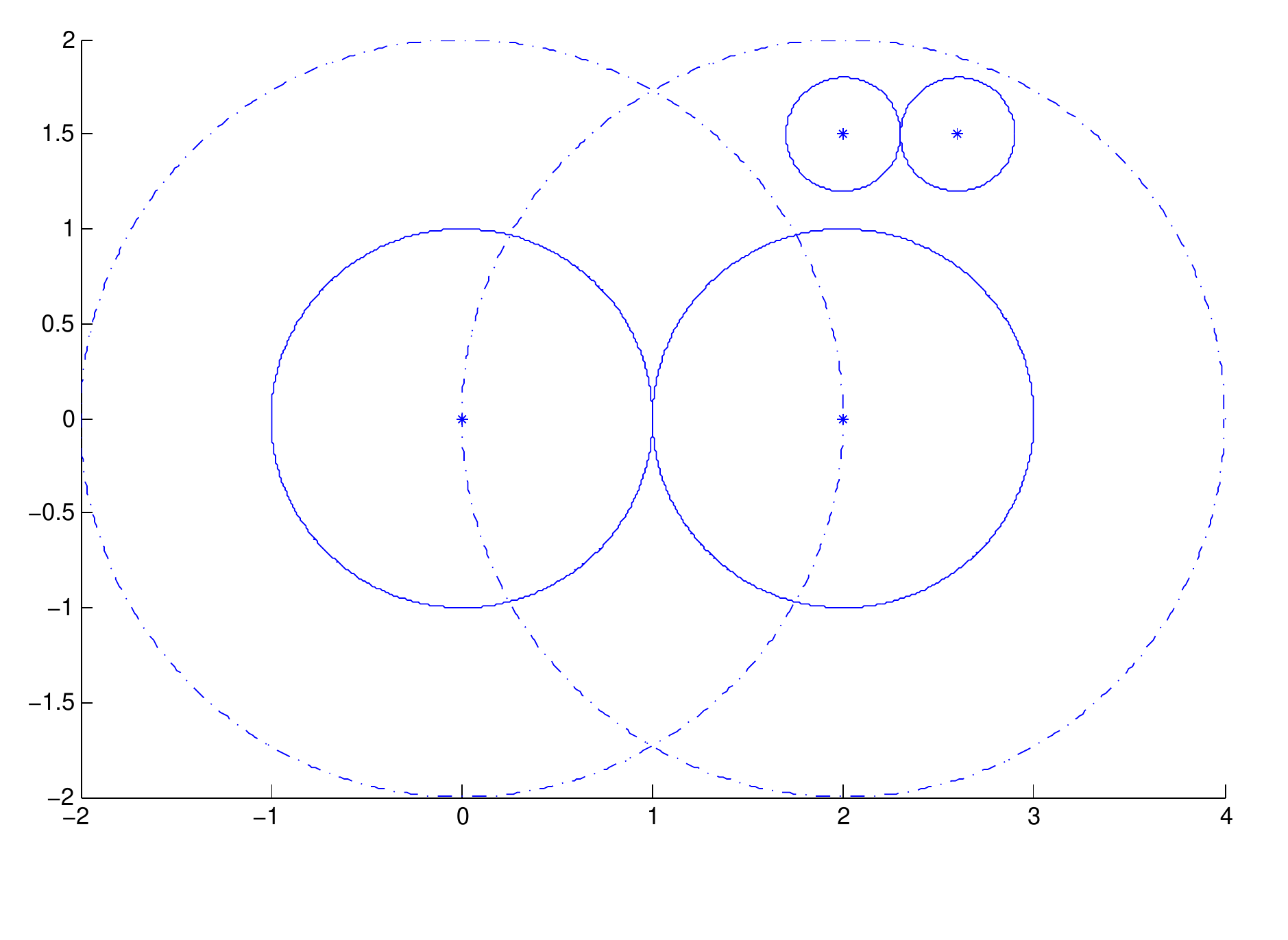}
	\caption{Illustrative explanation of the differences between the NN and the LL graph.}
\label{NNvsLLexemp}
\end{figure}

\section{Clusters in the nearest-neighbour (NN) graph}
\label{SecV}

\subsection{Properties and Results}

As discussed previously, all clusters $\mathcal{C}_m$ in the NN model are a.s. finite. This was presented in property \textbf{P.IV} and holds due to the fact that there exists no infinite component. We randomly choose an atom and set its position as the origin of the planar coordinates. This atom is called the \textit{typical atom} and is denoted by $\mathbf{z}_o$. Let us further denote by $\mathcal{C}_o$ the \textit{typical cluster}, or else the cluster in which the typical atom belongs to. Then, starting from the typical point we form the sequence of atoms $\mathbf{z}_o\stackrel{1}{\rightarrow}\mathbf{z}_1\stackrel{1}{\rightarrow}\ldots\mathbf{z}_{n}\stackrel{1}{\rightarrow}\mathbf{z}_{n+1}\stackrel{1}{\rightarrow}\ldots$, where the $(n+1)$-th atom is the nearest neighbour of the $n$-th one. Hence $\left\{\mathbf{z}_o,\mathbf{z}_1,\ldots\right\}\subseteq\mathcal{C}_o$. From the distances $r_n:=\left|\mathbf{z}_{n}-\mathbf{z}_{n+1}\right|$, we obtain the sequence $\left\{r_o,r_1,\ldots\right\}$. To fascilitate the analysis, we introduce the \textit{directed NN graph} $\vec{\mathcal{G}}_{NN}$ where $\mathbf{z}_i, \mathbf{z}_j$ are connected by an arc if $j$ is the nearest neighbour of $i$. By replacing the arcs by undirected edges we get $\mathcal{G}_{NN}$. The introduced sequence of atoms makes more sense in the directed graph.

In the NN model, a very important subset of atoms of $\phi$ are the \textit{mutually-nearest-neighbours}, which are pairs of atoms with the property that the one is the nearest neighbour of the other and vice-versa. Their importance is better understood by the following two Lemmas.

\begin{Lem}
\label{Lem:NNpair}
All clusters of the NN model contain exactly one pair of mutually-nearest-neighbours.
\end{Lem}

\begin{proof}
Suppose there is no pair of mutually-nearest-neighbours in a cluster, say the typical one $\mathcal{C}_o$. Then, the sequence of atoms $\left\{\mathbf{z}_o,\mathbf{z}_1,\ldots\right\}$ described above will either be infinite (impossible by property \textbf{P.IV}) or will be finite with $N<\infty$. In the second case, the nearest neighbour of  $\mathbf{z}_N$ is some atom previously encountered in the sequence with index $n\leq N-2$ and a cycle will appear in the cluster $\mathcal{C}_o$ of the graph $\vec{\mathcal{G}}_{NN}$ (and $\mathcal{G}_{NN}$), with size greater or equal to 3. This is impossible (see also \textbf{P.III}). 

To see why, choose without loss of generality (w.l.o.g.) $\mathbf{z}_N\stackrel{1}{\rightarrow}\mathbf{z}_{N-2}$. Then the open ball $\mathcal{B}\left(\mathbf{z}_N,r_N\right)$ with center $\mathbf{z}_N$ and radius $r_N:=\left|\mathbf{z}_{N}-\mathbf{z}_{N-2}\right|$ will be empty with $\mathbf{z}_{N-2}$ on its boundary. Around $\mathbf{z}_{N-2}$ there should be an empty open ball of radius $r_{N-2}$ with $\mathbf{z}_{N-1}$ on its boundary, and for this to hold $r_{N-2}<r_{N}$. Given that $\mathcal{B}\left(\mathbf{z}_{N},r_{N}\right)$ is empty, the atom $\mathbf{z}_{N-1}$ should lie outside this ball, hence at a distance $r_{N-1}>r_N$. However, this contradicts the fact that $\mathbf{z}_N$ is the nearest neighbour of $\mathbf{z}_{N-1}$, because the ball $\mathcal{B}\left(\mathbf{z}_{N-1},r_{N-1}\right)$ is not empty but contains $\mathbf{z}_{N-2}$, since $r_{N-1}>r_N$ and $r_{N-2}<r_{N}$.

We have proven existence of at least one mutually-nearest-neighbours pair within each cluster. The uniqueness comes from the fact that there exists no path of atoms in $\vec{\mathcal{G}}_{NN}$, that connects two distinct such pairs. Atoms of such pairs can only be nearest neighbours for some other atom.
\end{proof}
As a conclusion of the above Lemma, starting from a typical point and moving through the directed path of the graph $\vec{\mathcal{G}}_{NN}$, we get a sequence of $N+1$ points $\left\{\mathbf{z}_{o}, \mathbf{z}_{1},\ldots,\mathbf{z}_{N}\right\}$, with $N\geq 2$, from which the two last ones $\mathbf{z}_{N-1}, \mathbf{z}_{N}$ are mutually-nearest-neighbours. This is the only pair of atoms of the cluster $\mathcal{C}_o$, having such property.

\begin{Lem}
\label{Lem:NNpairD}
The distances $\left\{r_o,r_1,\ldots,r_N\right\}$ - starting from the typical atom $\mathbf{z}_o$ - form a strictly decreasing finite sequence until the $\left(N-1\right)$-th atom. The sequence ends with $r_{N-1}=r_N$. Hence, the minumum distance is the one between the mutually-nearest-neighbours. The atom $\mathbf{z}_N$ is called the "root" of the branch of the cluster, where $\mathbf{z}_o$ belongs to.
\end{Lem}

\begin{proof}
Starting from $\mathbf{z}_o$, the ball $\mathcal{B}\left(\mathbf{z}_o,r_o\right)$ is empty, with $\mathbf{z}_1$ on its boundary. Since these are not mutually-nearest-neighbours, the ball $\mathcal{B}\left(\mathbf{z}_1,r_1\right)$ is empty if $r_1<r_o$ and has $\mathbf{z}_2$ on its boundary, which should additionally lie outside the ball $\mathcal{B}\left(\mathbf{z}_o,r_o\right)$. We iterate this process until the ball of point $\mathbf{z}_{N-1}$ with $\mathbf{z}_{N}$ on its boundary. The sequence stops when $r_N=r_{N-1}$ and the nearest neighbour of $\mathbf{z}_{N}\stackrel{1}{\rightarrow}\mathbf{z}_{N-1}$. An example of such a sequence of decreasing empty balls is shown in Fig.\ref{NNexemp5}.
\end{proof}

\begin{figure}[h]
\centering
 \includegraphics[trim = 15mm 85mm 10mm 85mm, clip, width=0.8\textwidth]{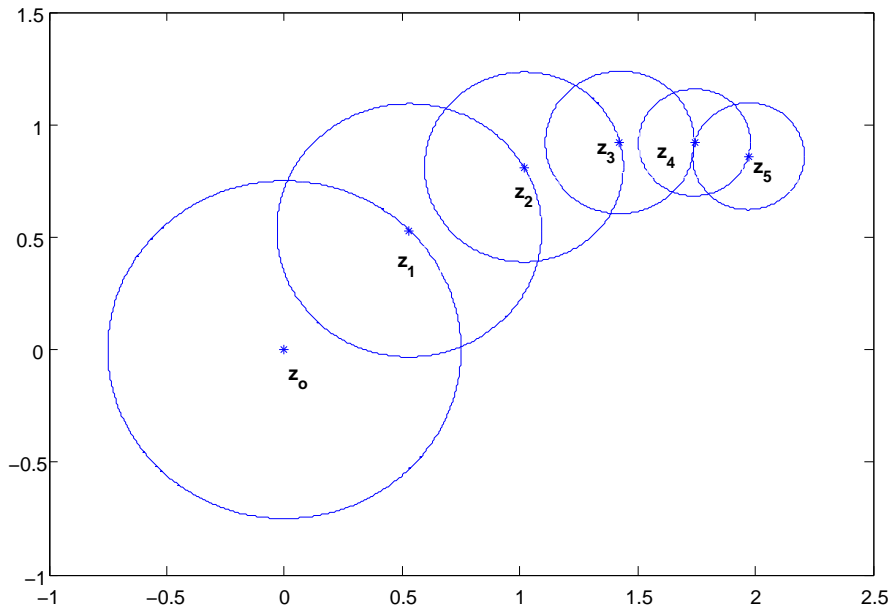}
	\caption{Illustrative presentation of a sequence of atoms $\left\{\mathbf{z}_o,\ldots,\mathbf{z}_5\right\}$ in the $\mathcal{G}_{NN}$ graph that have first neighbour relation. The figure shows the empty balls centered on each atom and the decreasing sequence of radii until the mutually-nearest-neighbour pair $\mathbf{z}_{N-1}\stackrel{1}{\leftrightarrow}\mathbf{z}_N$.}
\label{NNexemp5}
\end{figure}

The event that the \textit{size of a cluster} is very large may have small probability, but not zero. Note here that $card\left(\mathcal{C}_o\right)\geq N+1\geq 2$ in the NN graph. Regarding the asymptotic behaviour of the cluster size, Kozakova \textit{et al.} in \cite{Kozakova06} have derived bounds on its tail probability. In the case of the 2D plane of interest, we present the following result from their work.

\begin{Thm}[Theorem 1.2, \cite{Kozakova06}] Define $\rho\left(N\right)$ as the probability that there is a directed path in $\vec{\mathcal{G}}_{NN}$ through the typical atom at the origin, touching more than $N$ distinct atoms. Then, there exist constants $C_1,C_2,N_o\in\left(0,\infty\right)$ such that
\begin{eqnarray}
\label{Assympt1}
e^{-C_1N\log N}\leq \rho\left(N\right) \leq e^{-C_2N\log N}, & & \forall N\geq N_o.
\end{eqnarray}
\end{Thm}
Note that the elements of such path through $\mathbf{z}_o$ are a subset of $\mathcal{C}_o$. Their cardinality can be understood as the depth of the tree $\mathcal{C}_o$, since (as mentioned in \textbf{P.III}) $\mathcal{G}_{NN}$ is a forest. Based on the above theorem, the probability that the length of such path inside the typical cluster $\mathcal{C}_o$ of $\vec{\mathcal{G}}_{NN}$ exceeds $N$, decays superexponentially with the number of atoms $N$ in the path. Taking also into account the fact that the number of atoms with the same nearest neighbour is small, we conclude that the probability of very large clusters is trivial.

\subsection{The generation number}
\label{subGen}

Another argument \cite{Kozakova06} in favor of the small size of clusters in the NN model, comes from the use of the so called \textbf{generation number}. 
As already mentioned, any cluster of $\mathcal{G}_{NN}$ contains exactly one pair of mutually-nearest-neighbours. These atoms are given generation number 1. A Poisson atom $\mathbf{z}_n$ receives generation number $k\geq 2$, if the graph distance (number of hops) to the unique pair in its cluster is equal to $k-1$. The generation number is connected to the length $N$ of the chain of atoms $\left\{\mathbf{z}_o,\mathbf{z}_1,\ldots,\mathbf{z}_N\right\}$ with the relation $k=N$. To see this, if $\mathbf{z}_{N-1},\mathbf{z}_N$ are mutually-nearest-neighbours and take the generation number $1$, then $\mathbf{z}_{N-2}$ will take generation number 2, so that $\mathbf{z}_{N-k}$ will take generation number $k$. Obviously for $\mathbf{z}_{o}$, $N-k=0\Rightarrow N=k$.

Let us denote by $g\left(k\right)$ the probability that the typical atom has generation number $k$. We can calculate exactly the probability of generation number $1$, which corresponds to the probability that a random atom of the p.p.p. belongs to a mutually-nearest-neighbour pair. 

\begin{Thm}
\label{ThMutualNN}
The exact probability that the typical atom of a p.p.p. $\Phi$ with density $\lambda>0$ has a mutually-nearest neighbour (or that the typical atom has generation number $1$) is equal to 
\begin{eqnarray}
\label{06215}
\mathbb{P}\left[\mathbf{z}_o\ in\ pair\right]=g\left(1\right) & = & 0.6215.
\end{eqnarray}
\end{Thm}

\begin{proof}
The result has been first encountered in \cite{DalStoSto99} and we add its proof here for completeness. Given that the volume of the ball $\mathcal{B}\left(\mathbf{z}_o,r\right)$ is $V_o$ ($=\lambda\pi r^2$), we can write the exact above probability as follows
\begin{eqnarray}
\label{VoINT}
g\left(1\right) & = & \int_{0}^{\infty}e^{-V_o}e^{-\left(1-\alpha\right)V_o}dV_o,
\end{eqnarray}
where $\left(1-\alpha\right)V_o$ is the area $\mathcal{B}\left(\mathbf{z}_1,r\right)\setminus\mathcal{B}\left(\mathbf{z}_o,r\right)$ outside the ball of $\mathbf{z}_o$. (Notice that this area is upper and lower bouned by $V_o$ and $V_o/2$ respectively to give the bounds in the Theorem \ref{Thgenernum} below). Furthermore, $\alpha V_o$ is the area of overlap for the two balls of equal volume with centers at distance $r$. The constant is equal to $\alpha = \frac{2}{3}-\frac{\sqrt{3}}{2\pi}$. Having said this, the integral can be easily solved to give $g\left(1\right) = \frac{1}{2-\alpha}=0.6215$.
\end{proof}

The probability $g\left(k\right)$, $k\geq 1$, is upper- and lower-bounded in the following Theorem. The proof of the lower bound comes from \cite[Theorem 1.3]{Kozakova06}, whereas the upper bound is a novel contribution.

\begin{Thm}
\label{Thgenernum}
Given a planar p.p.p. $\Phi$, the probability that the generation number, i.e. the distance of the typical atom from the mutually-nearest-neighbours, is $k\in\left\{1,2,3,\ldots\right\}$ can be bounded as follows
\begin{eqnarray}
\label{BoundGNb}
2^{k}\frac{k+1}{\left(k+2\right)!}\geq g\left(k\right)\geq \frac{k}{\left(k+1\right)!}-\mathcal{Q}\left(k\right), & & k\geq 1.
\end{eqnarray}
The correction term $\mathcal{Q}\left(k\right):=q\left(k\right)+\sum_{i=1}^{k-1}\mathcal{Q}\left(i\right)- \mathcal{Q}\left(i-1\right)$, $\mathcal{Q}\left(0\right)=0$, $\mathcal{Q}\left(k\right)\geq 0$, is increasing in $k$. It holds $\mathcal{Q}\left(1\right) = q\left(1\right)=0$ and $\mathcal{Q}\left(k\right)\geq \mathcal{Q}\left(2\right) = \frac{\alpha}{2}$, where $\alpha$ is defined as previously. The lower bound expression is equal to zero for $k\geq 3$.
\end{Thm}

\begin{proof}
Consider a typical point $\mathbf{z}_o$ und suppose its generation number is $k$. Then we have a sequence $\left\{\mathbf{z}_o,\mathbf{z}_1,\ldots,\mathbf{z}_N\right\}$ of $N+1$ atoms, among which the last two constitute mutually-nearest-neighbours. From Lemma \ref{Lem:NNpairD} there is a sequence of balls (disks in 2D) with centers the above atoms and decreasing radii, as shown in Fig.\ref{NNexemp5} for six atoms. Suppose that the decreasing sequence of their volumes is $V_o\geq \ldots\geq V_{N-1}=V_N\geq 0$. The probability of such an event is:
\begin{itemize}
\item lower-bounded by the probability that the area $\mathcal{A}^{L}$, equal to the \textit{sum of all N+1 disks without overlapping}, is empty of any other atoms, reduced by a correction term $\mathcal{Q}\left(N\right)$, which is increasing over $N$. We calculate this by
\begin{eqnarray}
\label{GN:LB1}
\int_{V_o\geq\ldots\geq V_{N-2}\geq V_{N-1}\geq 0}e^{-V_o}e^{-V_1}\cdots e^{-V_{N-2}}e^{-2V_{N-1}}\ dV_{N-1}dV_{N-2}\cdots dV_1dV_{o} -\mathcal{Q}\left(N\right).
\end{eqnarray}
\item upper-bounded by the probability that the area $\mathcal{A}^{U}$, equal to the \textit{sum of the first disk $V_o$ and all the rest N semidisks without overlapping}, is empty of any other atoms and we calculate this by
\begin{eqnarray}
\label{GN:UB1}
\int_{V_o\geq\ldots\geq V_{N-2}\geq V_{N-1}\geq 0}e^{-V_o}e^{-V_1/2}\cdots e^{-V_{N-2}/2}e^{-V_{N-1}}\ dV_{N-1}dV_{N-2}\cdots dV_1dV_{o}.
\end{eqnarray}
\end{itemize}
The lower bound is obvious, since the area of all disks without overlapping is larger that the actual area $\mathcal{A}\subset\mathcal{A}^L$. To understand the term $\mathcal{Q}\left(N\right)$ we first provide the intuition for the case $N=2$. Given the typical atom $\mathbf{z}_o$, the atom $\mathbf{z}_1$ may lie anywhere on its boundary and is the center of a ball $\mathcal{B}\left(\mathbf{z}_1,r_1\right)$, with $r_1\leq r_0$. Furthermore, the atom $\mathbf{z}_2$, which is the second part of the mutually-nearest-neighbours pair, will lie on the boundary of $\mathcal{B}\left(\mathbf{z}_1,r_1\right)$. However the first integral in (\ref{GN:LB1}) does not take into account the fact that, the atom $\mathbf{z}_2$ should not lie inside the first ball $\mathcal{B}\left(\mathbf{z}_o,r_o\right)$. This probability should be subtracted from the integral and is shown as $\mathcal{Q}\left(2\right)$. Obviously $\mathcal{Q}\left(1\right) = 0$, because such an event can not be encountered with less than three atoms. 
To calculate the term for $N=2$, we use the bound in \cite[pp.10-11]{Kozakova06}
\begin{eqnarray}
\label{Q2}
\mathcal{Q}\left(2\right) & \leq & \int_{V_o\geq V_1\geq 0} \frac{V\left(\mathcal{B}\left(\mathbf{z}_1,r_o\right)\cap \mathcal{B}\left(\mathbf{z}_o,r_o\right)\right)}{V\left(\mathcal{B}\left(\mathbf{z}_1,r_o\right)\right)}e^{-V_o}e^{-V_1}\ dV_1 dV_o\nonumber\\
& = &  \int_{V_o\geq V_1\geq 0} \frac{a V_o}{V_o}e^{-V_o}e^{-V_1}\ dV_1 dV_o = \frac{\alpha}{2}.
\end{eqnarray}

For higher number of $N$, given that the previous atoms are not contained in the union of balls up to $N-1$, we calculate the correction term $q\left(N\right)$ for the atom with index $N$. However, notice that the integral in (\ref{GN:LB1}) is equal to $\frac{N}{\left(N+1\right)!}< \alpha/2$ for $N\geq 3$, so that no more terms should be calculated for the lower bound, because the expression in (\ref{GN:LB1}) is zero for $N\geq 3$.

To understand the upper bound, consider first the set of $N+1$ overlapping balls with centers the points of the sequence $\left\{\mathbf{z}_o,\mathbf{z}_1,\ldots,\mathbf{z}_N\right\}$ and radii $\left\{r_o,r_1,\ldots,r_N\right\}$. The total area $\mathcal{A}$ depends on the relative positions of the atoms. It is minimum, when the positions are such that the area of  overlapping is maximum. This happens when no three or more balls overlap. If - say - three balls overlap, we can always move the center of the third ball but keeping the radius fixed, so that the area of overlap between the third and the second (or the first) increases. Hence, the minimum area $\mathcal{A}$ results when all atoms fall on a line and have distance $r_n$, $n=0,\ldots,N$ between them. In such a linear constellation, consider the ball $\mathcal{B}_o:=\mathcal{B}\left(\mathbf{z}_o,r_o\right)$. This ball overlaps with $\mathcal{B}_1$ and touches $\mathbf{z}_1$ on its boundary. Obviously, the area of overlap cannot be more that $V_1/2$, leaving exactly the same area $V_1/2$ untouched. Iterating the process, we get $\mathcal{A}^{U}=V_o+V_1/2+\ldots+V_{N-1}/2+V_N/2\subset\mathcal{A}$, but considering that $V_{N-1}=V_N$ we reach the presented bound. The integrals can be easily solved by hand and derive the presented results.

We could simply calculate these nested integrals and be through. However, there is an easier way to find their closed form solution with the following probabilistic interpretation, inspired by \cite[p.10]{Kozakova06}. For the lower bound, the integral (\ref{GN:LB1}) is equal to 1/2 times ($V_{N-1}=V_N$) the probability that $V_o\geq \ldots\geq V_{N-1}\geq 0$, where the $V_i$'s are independent, $V_o,\ldots,V_{N-2}$ are exponentially distributed and $V_{N-1}$ is also exponentially distributed with parameter 2. The probability that $V_o,\ldots,V_{N-2}$ are ordered in this way is $1/(N-1)!$. Since $V_{N-1}$ has the same distribution as the minimum of two independent exponentially distributed random variables, the probability that one of these two will be the smallest among the $N+1$ (previous $N-1$ plus these 2) is $2/(N+1)$. It follows that the first integral is equal to
\begin{eqnarray}
\label{GN:INT1}
\frac{1}{2}\frac{1}{(N-1)!}\frac{2}{N+1} & = & \frac{N}{(N+1)!}.\nonumber
\end{eqnarray}
For the upper bound, we make the transformation $V_i=2U_i$ in the integral, which results in a term $2^{N}$ outside the integral. Consider first the case $N\geq 3$. All $N$ variables are independent exponentially distributed, but $U_o$ and $U_{N-1}$ have parameter 2 and can be considered each as the minimum of two independent exponentially distributed random variables with parameter 1. Suppose $U_{o,a}$ and $U_{o,b}$ are the auxiliary variables, so that $U_o=\min\left\{U_{o,a},U_{o,b}\right\}$ and suppose the same for $U_{N-1}$. The probability in question is $1/2$ (case $U_{o,a}=U_{o,b}$) times $1/2$ (case $U_{N-1,a}=U_{N-1,b}$) the probability that $U_{o,a},U_{o,b}\geq U_1\geq \ldots\geq U_{N-2}$ are ordered this way. The latter is $2/N!$, where the 2 comes since we are indifferent regarding the ordering of $U_{o,a},U_{o,b}$. Furthermore, the probability that one of the two $U_{N-1,a}$ or $U_{N-1,b}$ is the smallest among the $N+2$ (previous $N-2$ plus these 2 plus the 2 referring to $U_o$) is equal to $2/(N+2)$. Altogether, we get
\begin{eqnarray}
\label{GN:INT1}
\frac{1}{2}\frac{1}{2}2^N\frac{2}{N!}\frac{2}{N+2} & = & 2^{N}\frac{N+1}{(N+2)!},\ N\geq 3.\nonumber
\end{eqnarray}
For the case $N=2$, $V_o$ and $V_1$ are both independent exponentially distributed with parameter 1. The probability that $V_o\geq V_1$ is simply $1/2$. For $N=1$, the area $\mathcal{A}^{UB}\left(1\right)=V_o+\frac{1}{2}V_o$.
\end{proof}
Observe that the integrals in (\ref{GN:LB1}) and (\ref{GN:UB1}) consider as variables the volumes $V_i$, with unit density, hence $V_i=\pi r_i^2$. In the case of $\lambda\neq 1$, $V_i=\lambda \pi r_i^2$ and the results will be exactly the same, hence the bounds are density-invariant, exactly as stated for the model NN in \textbf{P.I}. We next present a numerical evaluation of the bounds in Table \ref{TabgkULB} and plot the results for comparison in Fig. \ref{FigGkplotB}. The table and figure, include values of the generation number probability from Monte Carlo simulations. The area simulated has dimensions $10\times 10$ [$m^2$] and the density of the simulated p.p.p. is $2$ [atoms/$m^2$] hence in average $\mathbb{E}\left[N_t\right]=200$ atoms are randomly placed per realization. Finally, the values of the generation number probability result after $5000$ iterations.

The procedure creates a $2\times N_t$ table per iteration ($N_t$ atoms). The atoms are enumerated in the first row and the first neighbour of each is given in the second. Starting from the first column, the table is searched until the root of the first cluster is found. Once this is done the generation number is equal to the number of hops until the root. The procedure continues to the search starting from the second entry and so on.

\begin{table}[h]
\caption{Numerical values for the bounds and Monte Carlo simulation of the generation number probability $g(k)$, for $k=1,\ldots,6$.}
\centering
\begin{tabular}{|l|l|l|l|l|l|l|l|l|l|l|}
\hline
			& $g(1)$ & $g(2)$	& $g(3)$	& $g(4)$	& $g(5)$ & $g(6)$\\\hline
Upper Bound & $\frac{2}{3}$ & $\frac{1}{2}$ & $\frac{4}{15}$ & $\frac{1}{9}$ & $\frac{4}{105}$ & $\frac{1}{90}$\\\hline
Monte Carlo & $0.6207$ &  $0.2756$ & $0.0815$ & $0.0183$ & $0.0033$ & $0.0006$\\\hline
Lower Bound & $\frac{1}{2}$ & $0.1378$ & $0$ & $0$ & $0$ & $0$\\
\hline
\end{tabular}
\label{TabgkULB}
\end{table}

\begin{figure}[h]
\centering
 \includegraphics[trim = 15mm 70mm 10mm 65mm, clip, width=0.7\textwidth]{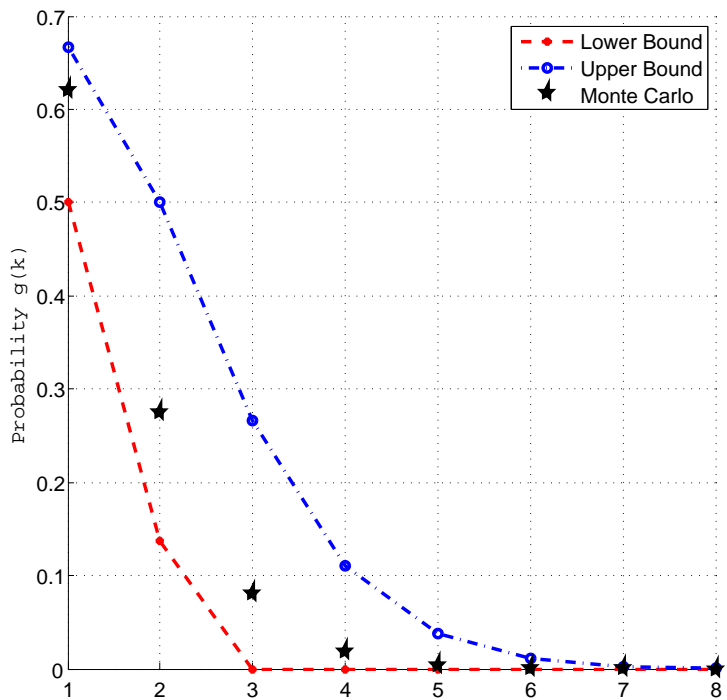}
	\caption{Evaluation of the upper and lower bounds of the generation number probability $g(k)$, for $k\geq 1$.}
\label{FigGkplotB}
\end{figure}

We can also check that the bounds derived in Theorem \ref{Thgenernum} are consistent with the result of Theorem \ref{ThMutualNN}, since $2/3\geq 0.6215\geq 1/2$. Also check how the Monte Carlo value $0.6207\approx 0.6215$.

An alternative way to interpret the generation number probability $g(k)$ is the following. 

\begin{Cor}
\label{Cor1}
The event that $\mathbf{z}_o$ has generation number $k$ implies that $\mathbf{z}_o$ belongs to a cluster (typical cluster $\mathcal{C}_o$) of cardinality at least $k+1$, but not the other way round. Consequently, the following relation holds
\begin{eqnarray}
\label{GNcardi}
\mathbb{P}\left[\mathbf{z}_o\in\mathcal{C}_o,\ card\left(\mathcal{C}_o\right)\geq k+1\right] & \geq & g\left(k\right) \geq \frac{k}{\left(k+1\right)!}-\mathcal{Q}\left(k\right),\ k\geq 1.
\end{eqnarray}
\end{Cor}

%---------------------------------

\section{The ancestor number}
\label{SecVI}
The generation number refers to exactly this branch of the cluster (tree), where the typical point is included. On each atom of the sequence $\left\{\mathbf{z}_o,\mathbf{z}_1,\ldots,\mathbf{z}_N\right\}$, other branches could be attached, whose number is limited by the maximum number of atoms having the same atom as nearest neighbour (often called \textit{kissing number} \cite{Kozakova06}, \cite{ZongKiss98}). In this sense, the generation number is appropriate for describing the directed graph $\vec{\mathcal{G}}_{NN}$, but is not necessarily a good indicator of the size of the typical cluster of $\mathcal{G}_{NN}$, so other measures may be more appropriate.

\subsection{Definition and Bounds}

\begin{Def}
\label{DefAncest}
Consider an atom of the p.p.p. $\mathbf{z}_n$ and suppose it belongs to cluster $\mathcal{C}_o$. Then, the \textbf{ancestor number} $k_a$ is defined iteratively:
\begin{enumerate}
\item The two atoms of the mutually-nearest-neighbour pair of $\mathcal{C}_o$ (see Lemma \ref{Lem:NNpair}) are given $k_a=1$ and form the set $\mathcal{S}_1$. By convention $\mathcal{S}_o$ contains just one of the two atoms of the pair, and initiates the process.
\item The atom $z_j\in\mathcal{C}_o$ which is closest to the set $\mathcal{S}_1$, is given $k_a=2$. Include the new atom in $\mathcal{S}_2=\mathcal{S}_1\cup\left\{\mathbf{z}_j\right\}$.
\item Find the closest atom in $\mathcal{C}_o$ to $\mathcal{S}_2$, which is given ancestor number $k_a=3$. Include it in the set $\mathcal{S}_3$.
\item ...
\end{enumerate}
In this way all $\mathbf{z}_n\in\mathcal{C}_o$ receive an ancestor number.
\end{Def}
Obviously, the ancestor number is increasing with the distance of the atom from the origin pair. By definition, an atom of the p.p.p. receives \textit{ancestor number} $k_a=K$, if and only if it is the closest neighbour of a set $\mathcal{S}_{K}$, and it itself has one of the members of $\mathcal{S}_{K}$ as nearest neighbour. When the latter is not true, the atom is not included in the cluster and $card\left(C_o\right)=K$. In this sense, the ancestor number best describes the cluster size, in terms of the count of its members. The furthest atom of a cluster to the origin pair, takes ancestor number equal to the cluster size.

Let us now relate the ancestor number to the \textit{generation number} by means of an example. Assume that the cluster $\mathcal{C}_o$, has $n$ branches stemming from the pair and hence, $n$ atoms with generation number $2$. Among these $n$, the atom with the shortest distance from (one of the two atoms of) the pair receives $k_a=2$. The ancestor number $k_a=3$ is given to the atom with minimum distance from the new set $\mathcal{S}_2$, which includes the two atoms with $k_a=1$ and the one with $k_a=2$. Observe that, this may be either one of the $n-1$ remaining atoms with generation number $2$, or the atom with generation number $3$, following the atom with $k_a=2$ on its branch. This explanation emphasizes the differences between the two numbers (\textit{generation} and \textit{ancestor}).

Based on the ancestor number, given a realization $\phi$ of a p.p.p. we can provide an \textbf{algorithm} to form NN clusters of increasing size:
\begin{enumerate}
\item Identify the mutually-nearest-neighbour pairs of the realization. These constitute the "roots" of all clusters $\mathcal{C}_m$ in the network and have $k_a=1$. This is done by finding the pairs of atoms $\mathbf{z}_i,\mathbf{z}_j$ with the property that the ball $\mathcal{B}\left(\mathbf{z}_i,R_1\right)$, with $R_1=d_{ij}:=\left|\mathbf{z}_i-\mathbf{z}_j\right|$ is empty and has atom $j$ on its boundary and the other way round.
\item Let the balls $\mathcal{B}_i$ and $\mathcal{B}_j$ grow further, with common radius, until their union meets an atom on its boundary, say $\mathbf{z}_k$. The radius is now $R_2=\min\left\{d_{ik},d_{jk}\right\}$. This candidate to be included in the cluster with $k_a=2$, should also have one of the two atoms in $\mathcal{S}_1$ as its nearest neighbour and for this, the ball $\mathcal{B}\left(\mathbf{z}_k,R_2\right)$ should be empty. If true we continue, else we stop and $card\left(\mathcal{C}_o\right)=2$.
\item Let the balls $\mathcal{B}_i$, $\mathcal{B}_j$, $\mathcal{B}_k$ grow further, until their union meets an atom on its boundary, say $\mathbf{z}_l$. The radius is $R_3=\min\left\{d_{il},d_{jl}, d_{kl}\right\}$. This atom takes the ancestor number $k_a=3$, if one of the three atoms in $\mathcal{S}_2$ are its nearest neighbour, in other words, the ball $\mathcal{B}\left(\mathbf{z}_l,R_3\right)$ is empty. etc.
\item ...
\end{enumerate}
How the algorithm gradually finds the elements of a cluster until ancestor number $k_a=3$, are illustrated in Fig.\ref{fig:algoA0}-\ref{fig:algoA3}. By letting the radius $R$ of the balls grow large enough, the algorithm will eventually converge to the $\mathcal{G}_{NN}$ graph. For a cluster with $card\left(\mathcal{C}_o\right)=K+1$, an empty ball of \textbf{the same} radius $R_K$ is centered on each of its atoms. Hence the probability of such cluster to appear, is equal to the probability that this union of $K+1$ balls is empty

\begin{eqnarray}
\label{UnionKempty}
\mathcal{A} := \bigcup_{k_a=1,1,2,\ldots,K}\mathcal{B}\left(\mathbf{z}_{k_a},R_K\right) & = & \emptyset.
\end{eqnarray}
Finally we note here, that the way the ancestor number is defined, gives rise to a sequence of \textit{ancestor radii} $\left\{R_1,\ldots,R_K\right\}$ for a cluster of size $K+1$. Each distance $R_n$ is the minimum distance of the atom with ancestor number $k_a=n$, to the set of its $n$ ancestors $\mathcal{S}_{n-1}$.

\begin{figure}[ht]    
\centering  
\label{fig:COOP1}
	 		\subfigure[Initial Configuration.]{          
           \includegraphics[trim = 5mm 70mm 5mm 70mm, clip, width=0.40\textwidth]{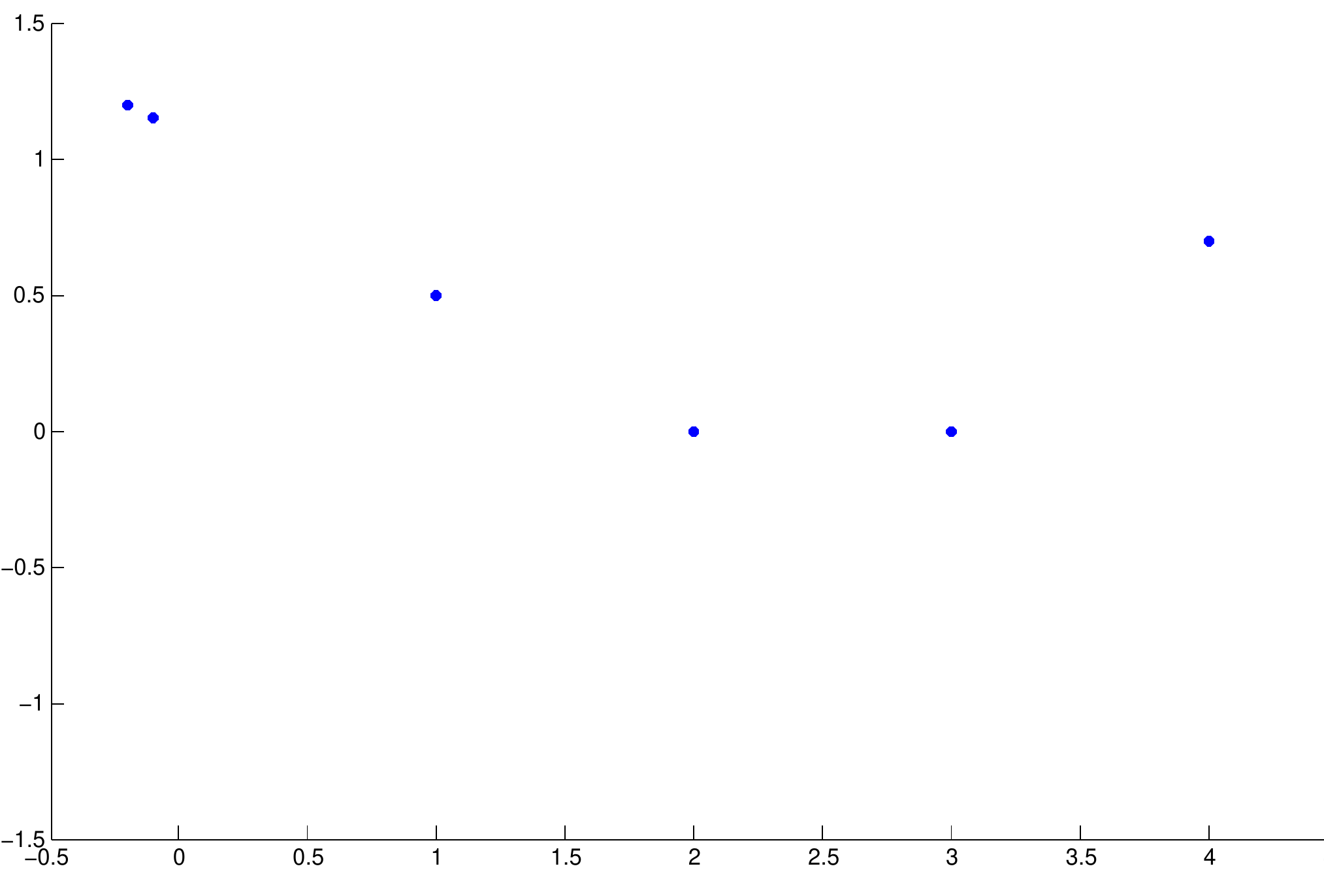}
           \label{fig:algoA0}
           }
            \subfigure[Algorithmic Step 1: pair $k_a=1$.]{          
           \includegraphics[trim = 10mm 80mm 10mm 80mm, clip, width=0.46\textwidth]{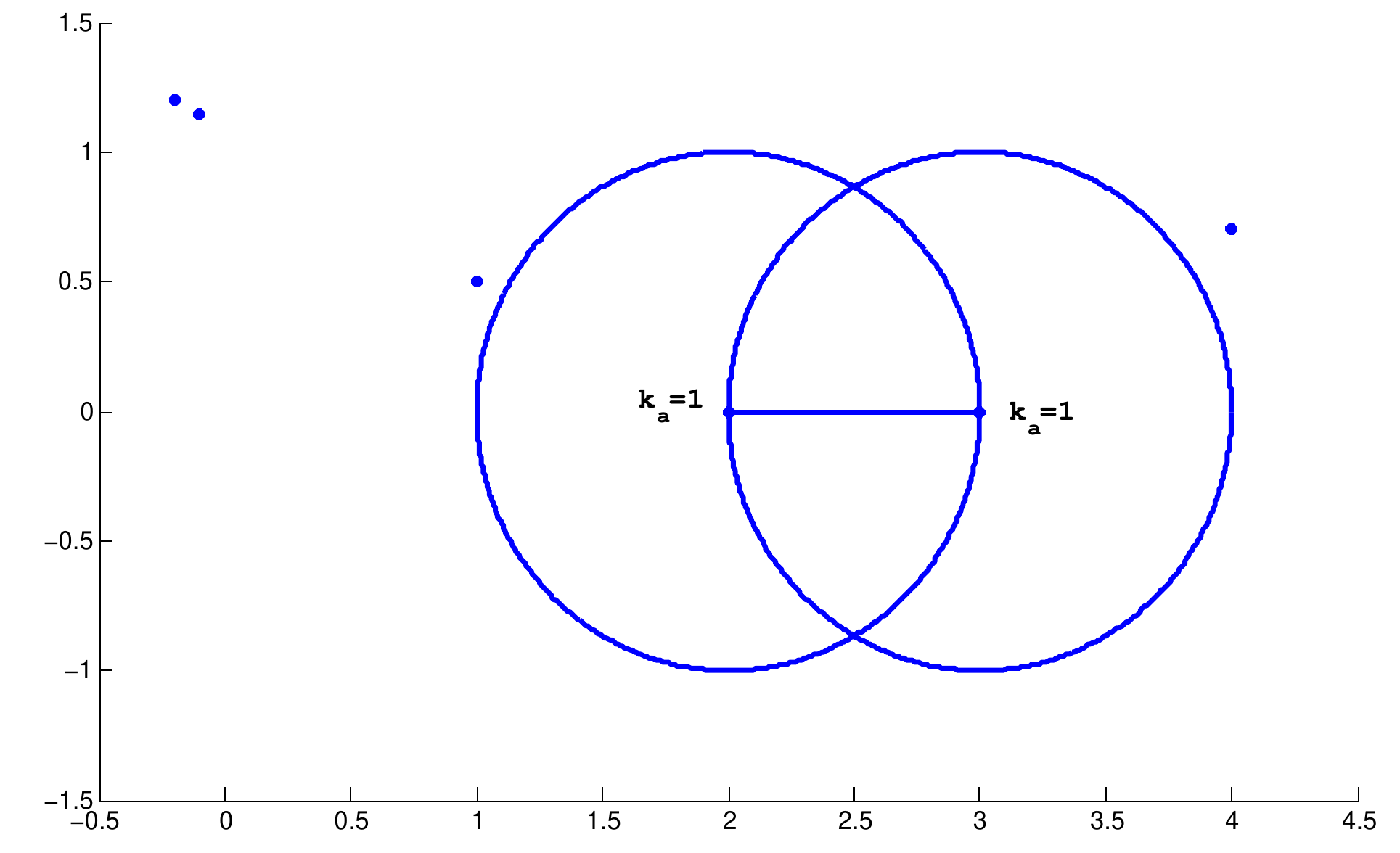}
           \label{fig:algoA1}
           }
            \subfigure[Algorithmic Step 2: $k_a=2$.]{          
           \includegraphics[trim = 10mm 75mm 10mm 70mm, clip, width=0.42\textwidth]{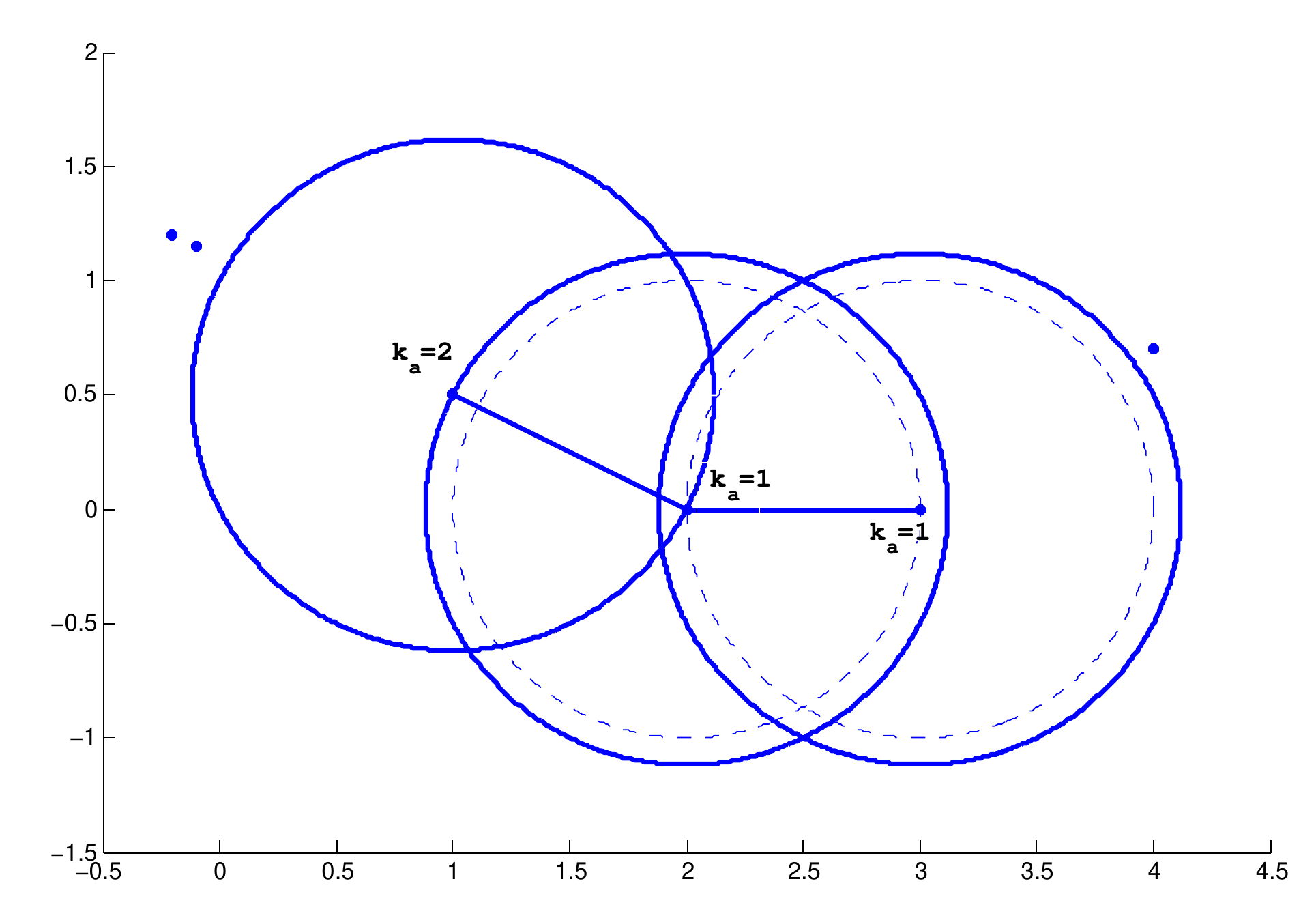}
           \label{fig:algoA2}
           }
            \subfigure[Algorithmic Step 3: $k_a=3$.]{          
           \includegraphics[trim = 10mm 85mm 10mm 70mm, clip, width=0.46\textwidth]{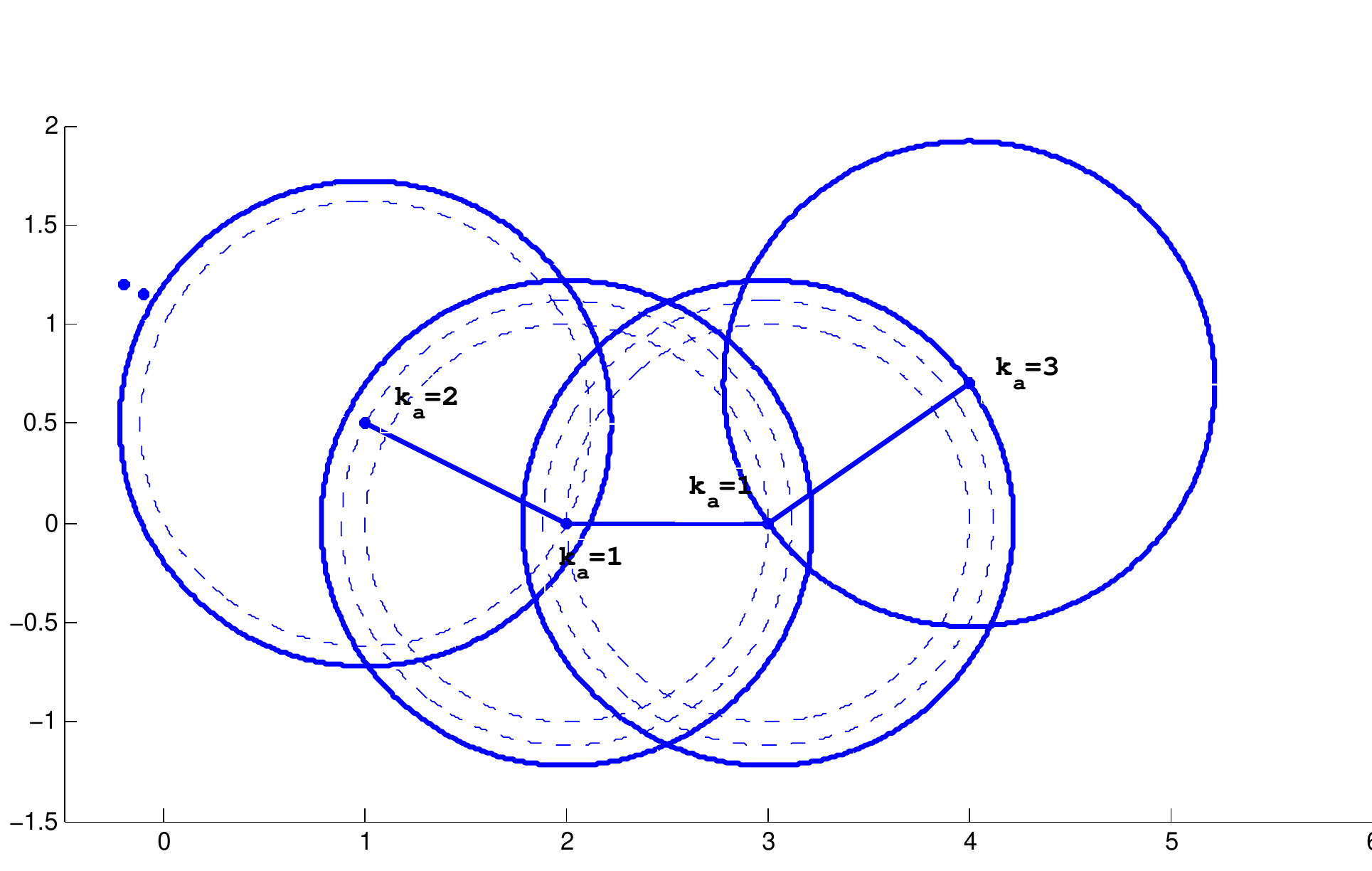}
           \label{fig:algoA3}
           }
           \caption{Implementation of the algorithm to form NN clusters based on the ancestor number. In the example, the algorithm forms a cluster of $4$ elements (atoms). The dashed circles in subfigures (c) and (d) illustrate the size of the circles in the previous steps of the algorithm. In subfigure (d) we can identify $2$ different branches connected to the origin pair, each one of which has a depth of $1$ atom. By further enlarging the area of (d) the candidate for $k_a=4$ would be met, but it is in nearest neighbour relation with a different atom than those of the cluster (the one exactly at its left). This is why the algorithm stops at $k_a=3$ and the cluster has cardinality $4$.}
\end{figure}

The probability that a typical point $\mathbf{z}_o$ has ancestor number $k_a$ is denoted by $g_a\left(k_a\right)$. In order to derive certain bounds, we need the following lemmas.  Let us denote by \textit{root}, one of the two atoms of the mutually-nearest-neighbour pair. We will make use of the \textit{nearest-neighbour-to-root radii} $\left\{\tilde{r}_1,\ldots,\tilde{r}_K\right\}$ which is the sequence of distances of the $n$-th nearest neighbour to the root with $1\leq n\leq K$, the sequence of \textit{ancestor radii} $\left\{R_1,\ldots,R_K\right\}$ defined above, and the sequence of \textit{auxiliary radii} $\left\{R\acute{}_1,\ldots,R\acute{}_K\right\}$ to be defined in Lemma \ref{LemmaAuxK}. We remind the reader here, that a r.v. $Y$ stochastically dominates $X$, and we write $X\preceq_{st} Y$ if $\mathbb{P}\left[X>t\right]\leq \mathbb{P}\left[Y>t\right]$, for all $t$.

\begin{Lem}
\label{LemmaNeighK}
The vector of $K$ \textit{nearest-neighbour-to-root radii} $\left\{\tilde{r}_1,\ldots,\tilde{r}_K\right\}$ stochastically dominates the vector of $K$ \textit{ancestor radii} $\mathbf{R} = \left(R_1,\ldots,R_K\right)$ in the sense that, 
\begin{eqnarray}
\label{stochdomANC1}
R_{k_a}\preceq_{st} \tilde{r}_{k_a}, & 1\leq k_a\leq K.
\end{eqnarray}
The joint probability distribution (p.d.f.) of the distances of $K$ nearest neighbours to the root, given that $\tilde{V}_k=\lambda \pi \tilde{r}_k^2$, $k\in\left\{1,\ldots,K\right\}$ is 
\begin{eqnarray}
\label{pdfneighK}
f_{\tilde{V}_1,\ldots, \tilde{V}_K}\left(\tilde{V}_1,\ldots, \tilde{V}_K\right) = e^{-\tilde{V}_K}, & \tilde{V}_K\geq\ldots \geq \tilde{V}_1 \geq 0.
\end{eqnarray}
\end{Lem}

\begin{proof}
The domination result, can be derived, by first exploring the case $k_a=2$ and applying similar arguments to higher ancestor numbers. Suppose the mutually-nearest-neighbour pair consists of atoms $\mathbf{z}_i,\mathbf{z}_j$. Obviously $R_1=\tilde{r}_1$. From the second step of the algorithm, the two balls around these atoms should expand their volume with common radius, until they touch $\mathbf{z}_k$, which is the candidate for $k_a=2$. This will occur at a distance $R_2\leq \tilde{r}_2$, because $\mathbf{z}_k$, might be the second neighbour of $\mathbf{z}_j$ instead of the root atom $\mathbf{z}_i$. In this sense, the probability that $\mathbf{z}_k$ is met at $R_2=\tilde{r}_2=t$, when $R_1=\tilde{r}_1=s$ can be calculated for both cases to be $e^{-V\left(\mathcal{B}_i\left(t\right)\cup\mathcal{B}_j\left(t\right)\right)+V\left(\mathcal{B}_i\left(s\right)\cup\mathcal{B}_j\left(s\right)\right)}<e^{-V\left(\mathcal{B}_i\left(t\right)\right)+V\left(\mathcal{B}_j\left(s\right)\right)}$ and as a result $\tilde{r}_2$ stochastically dominates $R_2$.

To further derive the p.d.f. for each $n$, we should first calculate the conditional probability that $\tilde{V}_K\leq t$, given that $\tilde{V}_{K-1}=s$, with $t\geq s$. In the p.p.p. case, this is equal to 
\begin{eqnarray}
\label{KcondK-1}
\mathbb{P}\left[\tilde{V}_K\leq t|\tilde{V}_{K-1}=s\right] & = & 1 - e^{-\left(t-s\right)}\Rightarrow\nonumber\\
f_{\tilde{V}_K|\tilde{V}_{K-1}}\left(t|s\right) & = & e^{-\left(t-s\right)}
\end{eqnarray}
Obviously by unconditioning iteratively from $\tilde{V}_{K-1}$ to $\tilde{V}_1$ and using (\ref{KcondK-1}) with appropriate replacement of the index, we reach the result. The case for $K=2$ has already been shown and used in \cite{BaccGiovAsilomar13} and \cite{AGFBarXiv13}.
\end{proof}

\begin{Lem}
\label{LemmaAuxK}
The vector of $K$ \textit{ancestor radii}  $\mathbf{R} = \left(R_1,\ldots,R_K\right)$ stochastically dominates the auxiliary $K$-length vector $\mathbf{R\acute{}} = \left(R\acute{}_1,\ldots,R\acute{}_K\right)$ in the sense that, 
\begin{eqnarray}
\label{stochdomANC2}
R\acute{}_{k_a}\preceq_{st} R_{k_a}, & 1\leq k_a\leq K,
\end{eqnarray}
The vector of auxiliary random variables (r.v.'s) $\mathbf{R\acute{}}$ has joint p.d.f. (with $V\acute{}_{k}=\lambda\pi R\acute{}_k^2$)
\begin{eqnarray}
\label{pdfANC}
f_{V\acute{}_1,\ldots, V\acute{}_K}\left(V\acute{}_1,\ldots,V\acute{}_K\right) = K!\ e^{-KV\acute{}_K}e^{V\acute{}_{K-1}}\cdots e^{V\acute{}_{1}}, & V\acute{}_K\geq\ldots \geq V\acute{}_1 \geq 0. 
\end{eqnarray}
\end{Lem}

\begin{proof}
The auxiliary vector $\mathbf{R}\acute{}$ is defined iteratively as follows. Let the first element $R\acute{}_1$ be the distance between the two atoms $\mathbf{z}_i$ (root), $\mathbf{z}_j$ of the pair, exactly as the ancestor distance, and $R\acute{}_1:=R_1$. For analytical purposes also $R\acute{}_0:=0$. Ancestor distance $R_2$ is explained in Step 2 of the previous algorithm to be the closest distance of an atom $\mathbf{z}_k$ to one of the two atoms $\mathbf{z}_i, \mathbf{z}_j$. To find this geometrically, the balls $\mathcal{B}_i, \mathcal{B}_j$ should grow with common radius until they meet $\mathbf{z}_k$ on their common boundary. Thus, $\mathbf{z}_k$ is the candidate for the ancestor number $k_a=2$. The probability that $R_2>t$, given that $R_1=s$ is equal to $e^{-V\left(\mathcal{B}_i\left(t\right)\cup \mathcal{B}_j\left(t\right)\right)+V\left(\mathcal{B}_i\left(s\right)\cup \mathcal{B}_j\left(s\right)\right)}$ for the p.p.p. case. Let the auxiliary random variable $R\acute{}_2$ be the common radius of the balls $\mathcal{B}\acute{}_i$, $\mathcal{B}\acute{}_j$ when one of them meets $\mathbf{z}_k\acute{}$ on its boundary, but in this case the one ball is "far away" from the other, hence there is no overlap. Obviously, this probability $e^{-V\left(\mathcal{B}\acute{}_i\left(t\right)\right)-V\left(\mathcal{B}\acute{}_j\left(t\right)\right)+V\left(\mathcal{B}\acute{}_i\left(s\right)\right)+V\left(\mathcal{B}\acute{}_j\left(s\right)\right)}<e^{-V\left(\mathcal{B}_i\left(t\right)\cup \mathcal{B}_j\left(t\right)\right)+V\left(\mathcal{B}_i\left(s\right)\cup \mathcal{B}_j\left(s\right)\right)}$ when $R_2=R\acute{}_2=t$ and $R_1=R\acute{}_1=s$, hence $R\acute{}_2$ is stochastically dominated by $R_2$. Repeating this argument stepwise for $3,\ldots,K$, we reach the domination result in (\ref{stochdomANC2}).

We enumerate the atoms by their ancestor number, giving to the root the number $0$. Let $V\acute{}_{n,K}$, $n=0,\ldots,K-1$ be the volume of the ball centered at $\mathbf{z}_n$, when the atom with ancestor number $K$ is touched by one of the $K$ balls (ancestors with balls $\mathcal{B}_0,\ldots,\mathcal{B}_{K-1}$). To find the joint p.d.f. of $\mathbf{R\acute{}}$ we first need to write the conditional probability 
\begin{eqnarray}
\label{JointpdfRacc}
& & \mathbb{P}\left[\min\left\{V\acute{}_{0,K},\ldots,V\acute{}_{K-1,K}\right\}\leq t|V\acute{}_{0,K-1}=\ldots =V\acute{}_{K-1,K-1}=s\leq t\right]  \nonumber\\
& \stackrel{(a)}{=} & 1-\left(\mathbb{P}\left[V\acute{}_{0,K}> t|V\acute{}_{0,K-1}=s\leq t\right]\right)^{K} \nonumber\\
& \stackrel{p.p.p.\ def}{=} & 1-\left(e^{-t+s}\right)^{K},\ K\geq 1.
\end{eqnarray}
To better clarify the formula, we give here two examples. For $K=1$, $\mathbb{P}\left[V\acute{}_{0,1} \leq t|V_{0}=0\right]=1-e^{-t}$ because per definition (above) $R\acute{}_0=0$. For $K=2$, $\mathbb{P}\left[\min\left\{V\acute{}_{0,2},V\acute{}_{1,2}\right\} \leq t|V_{0,1}=V_{1,1}=s\right]=1-e^{-2t+2s}$ (see also Fig.\ref{fig:algoA2} for an illustrative explanation). In the equations, (a) comes from the fact that the $K$ events of a ball centered at $\mathbf{z}_n$  with radius $R\acute{}_K$ having an empty ring in $\left[R\acute{}_K,R\acute{}_{K-1}\right]$ (i.e. the volume $V\acute{}_{n,K}-V\acute{}_{n,K-1}$, $n=0,\ldots,K-1$ is empty of atoms) are i.i.d. because the p.p.p. is homogenous and the balls are considered "far away" from each other.
The joint p.d.f. in (\ref{pdfANC}) results by first differentiating over $t$
\begin{eqnarray}
\label{JointpdfRacc2}
f_{V\acute{}_{0,K}|V\acute{}_{0,K-1}} & = & Ke^{-Kt+Ks},
\end{eqnarray}
and further applying Bayes' rule with the product
\begin{eqnarray}
\label{JointpdfRacc3}
\prod^{K}_{n=1}f_{V\acute{}_{0,n}|V\acute{}_{0,n-1}} & = & \prod_{n=1}^K ne^{-ns_n+ns_{n-1}},\ \ s_0=0.\nonumber
\end{eqnarray}
\end{proof}

\begin{Thm}
\label{AncestorBoundsUD}
Given a planar p.p.p. $\Phi$ the probability $g_a\left(k_a\right)$ that the typical atom $\mathbf{z}_o$ has ancestor number $k_a$ can be bounded as follows
\begin{eqnarray}
\label{AB2}
%\frac{k_a!}{\left(k_a + 1-\alpha\right)\left(k_a-1 + 1-\alpha\right)\cdots\left(1 + 1-\alpha\right)}
\frac{k_a!}{\prod_{n=0}^{k_a-1}\left(k_a+1-\alpha-n\left(1-\frac{1}{\pi}\right)\right)}\geq g_a\left(k_a\right)\geq \frac{1}{\left(1+k_a\left(1-\alpha\right)\right)^{k_a}}-\tilde{\mathcal{Q}}\left(k_a\right), & & k_a\geq 1.
\end{eqnarray}
The correction term $\tilde{\mathcal{Q}}\left(k_a\right):=\tilde{q}\left(k_a\right)+\sum_{i=1}^{k_a-1}\tilde{\mathcal{Q}}\left(i\right)- \tilde{\mathcal{Q}}\left(i-1\right)$, $\tilde{\mathcal{Q}}\left(0\right)=0$, $\tilde{\mathcal{Q}}\left(k_a\right)\geq 0$, is increasing in $k_a$. It holds $\tilde{\mathcal{Q}}\left(1\right) = \tilde{q}\left(1\right)=0$ and $\tilde{\mathcal{Q}}\left(k_a\right)\geq \tilde{\mathcal{Q}}\left(2\right) = 0.0795$. The lower bound expression is equal to zero for $k_a\geq 3$.
The constant $\alpha=\frac{2}{3}-\frac{\sqrt{3}}{2\pi}$.
\end{Thm}

\begin{proof}
For the lower bound we first choose some radius $R_{k_a}=t$. Then the volume of the area $\mathcal{A}^{LB}\left(t\right)$
\begin{eqnarray} 
\label{VolumeLB}
V\left(\mathcal{A}^{LB}\left(t\right)\right) 	& = & k_a \left(1-\alpha\right)V_0\left(t\right) \geq V\left(\bigcup_{n=1}^{k_a} \mathcal{B}_n\left(t\right)\setminus \mathcal{B}_{n-1}\left(t\right)\right) \stackrel{(\ref{UnionKempty})}{=} V\left(\mathcal{A}\setminus\mathcal{B}_0\left(t\right)\right),\nonumber
\end{eqnarray}
because the distances between atoms of the cluster are at most $t$. We further make use of Lemma \ref{LemmaNeighK} which states that $R_{k_a}\preceq_{st} \tilde{r}_{k_a}$. As a result the area $\mathcal{A}^{LB}$ with radius $\tilde{r}_{k_a}$ is stochastically larger than the same area with radius $R_{k_a}$. The bound is further calculated using the p.d.f. in (\ref{pdfneighK})

\begin{eqnarray}
\label{LB2}
\int_0^{\infty}d\tilde{V}_1\cdots\int_{\tilde{V}_{k_a-2}}^{\infty}d\tilde{V}_{k_a-1}\int_{\tilde{V}_{k_a-1}}^{\infty}d\tilde{V}_{k_a} e^{-\tilde{V}_{k_a}}e^{-k_a\left(1-\alpha\right)\tilde{V}_{k_a}} - \tilde{\mathcal{Q}}\left(k_a\right).
\end{eqnarray}
In the above, $\tilde{\mathcal{Q}}\left(k_a\right)$ is a correction term, which is non-decreasing in $k_a$. To understand this, consider the case $k_a=2$ in Fig. \ref{fig:algoA2}. Suppose that the atom with $k_a=2$ is the typical atom $\mathbf{z}_2:=\mathbf{z}_o$, with distance $R_2$ from its first neighbour ($\mathbf{z}_1$) with $k_a=1$ (the closest of the two roots). The ball $\mathcal{B}\left(\mathbf{z}_o,R_2\right)$ is empty and the atom $\mathbf{z}_1$ lies on its boundary. However this is in mutually-nearest-neighbour relation with $\mathbf{z}_0$, which lies on a distance $R_1$ at the boundary of the ball $\mathcal{B}\left(\mathbf{z}_1,R_1\right)$. It should necessarily lie outside the ball $\mathcal{B}\left(\mathbf{z}_o,R_2\right)$, so the integral on the left-hand side of (\ref{LB2}) should not consider such points for $\mathbf{z}_0$. Having said this, the correction term can be calculated using the inequality

\begin{eqnarray}
\label{CorrAnc}
\int_{\mathbf{z}_0\in\mathcal{B}\left(\mathbf{z}_2,R_2\right)\cap \mathcal{B}\left(\mathbf{z}_0,R_1\right)}f\left(\mathbf{z}_0\right)d\mathbf{z}_0 & \leq &  \frac{V\left(\mathcal{B}\left(\mathbf{z}_2,R_2\right)\cap \mathcal{B}\left(\mathbf{z}_0,R_1\right)\right)}{V\left(\mathcal{B}\left(\mathbf{z}_2,R_2\right)\right)} \int_{\mathbf{z}_0\in\mathcal{B}\left(\mathbf{z}_2,R_2\right)}f\left(\mathbf{z}_0\right)d\mathbf{z}_0\nonumber\\
& \leq & \alpha \int_{\mathbf{z}_0\in\mathcal{B}\left(\mathbf{z}_2,R_2\right)}f\left(\mathbf{z}_0\right)d\mathbf{z}_0.
\end{eqnarray}
From the above, the correction term for $k_a=2$ is 

\begin{eqnarray}
\label{CorrTerm2}
\tilde{\mathcal{Q}}\left(2\right) = \tilde{q}\left(2\right) & \leq & \alpha\int_0^{\infty}d\tilde{V}_1\int_{\tilde{V}_1}^{\infty}d\tilde{V}_2 e^{-\tilde{V}_2}e^{-2(1-\alpha)\tilde{V}_2}\nonumber\\
& = & \frac{\alpha}{\left(1+2(1-\alpha)\right)^2} = 0.0795.
\end{eqnarray}
The correction for $k_a=3$ will use this term and add the extra $\tilde{q}\left(3\right)$, which comes for the case that the point $\mathbf{z}_0$ falls in the union of balls $\mathcal{B}\left(\mathbf{z}_3,R_3\right)\cup \mathcal{B}\left(\mathbf{z}_2,R_2\right)$, given that $\mathbf{z}_1$, does not fall in the ball $\mathcal{B}\left(\mathbf{z}_3,R_3\right)$ ($=\tilde{q}(2)$). This means that the correction term $\tilde{\mathcal{Q}}\left(k_a\right)\geq \tilde{\mathcal{Q}}\left(2\right)$, for $k_a\geq 2$. As a result the lower bound is $0$ for $k_a\geq 3$.

For the upper bound we again choose some radius $R_{k_a}=t$. Then the volume of the area $\mathcal{A}^{UB}\left(t\right)$
\begin{eqnarray} 
\label{VolumeLB}
V\left(\mathcal{A}^{UB}\left(t\right)\right) 	& = & \left(1-\alpha\right)V_0\left(t\right) \leq V\left(\bigcup_{n=1}^{k_a} \mathcal{B}_n\left(t\right)\setminus \mathcal{B}_{n-1}\left(t\right)\right) \stackrel{(\ref{UnionKempty})}{=} V\left(\mathcal{A}\setminus\mathcal{B}_{0}\left(t\right)\right),\nonumber
\end{eqnarray}
because the area at the left-hand side is exactly the area $\mathcal{B}_{k_a-1}\left(t\right)\setminus \mathcal{B}_{k_a}\left(t\right)$ (equal to the symmetric $\mathcal{B}_{k_a}\left(t\right)\setminus \mathcal{B}_{k_a-1}\left(t\right)$) and does not contain any part of the rest $k_a-1$ planar subsets $\mathcal{B}_n\left(t\right)\setminus \mathcal{B}_{n-1}\left(t\right)$. Including part of these further tightens the bound. Observe that all the balls have the same radius $t$ and distances between them equal to $R_n\leq t$, $1\leq n\leq k_a$. Obviously, when the ancestor number is $k_a$, then $t=R_{k_a}$. The area of $\mathcal{B}_n\left(t\right)\setminus \mathcal{B}_{n-1}\left(t\right)$ can be given as a result of the area of overlap of two circles with the same radius $t=R_{k_a}$ and distance between centers $R_n$. The area in question is equal to \cite{Wolfr2circles}

\begin{eqnarray}
\label{SetminusTwoBalls}
\mathcal{B}_n\left(R_{k_a}\right)\setminus \mathcal{B}_{n-1}\left(R_{k_a}\right) & = & \pi R_{k_a}^2\left( 1-\frac{2}{\pi}\arccos\left(\frac{R_n}{2R_{k_a}}\right)+\frac{R_n}{2\pi R_{k_a}}\sqrt{4-\frac{R_n^2}{R_{k_a}^2}}\right).
\end{eqnarray}
Observe that for $R_n=R_{k_a}$, the expression above reduces to $\pi R_{k_a}\left(1-\alpha\right)$. We use the first order Maclaurin series approximation of the $\arccos$ function to bound the area, which gives $\arccos\left(x\right)\leq \frac{\pi}{2} - x$. As a result, the total area 
\begin{eqnarray}
\label{SetminusAllBalls}
\sum_{n=1}^{k_a}\mathcal{B}_n\left(R_{k_a}\right)\setminus \mathcal{B}_{n-1}\left(R_{k_a}\right) & = & V\left(\mathcal{B}_{k_a}\left(R_{k_a}\right)\setminus \mathcal{B}_{k_a-1}\left(R_{k_a}\right)\right) + \sum_{n=1}^{k_a-1}V\left(\mathcal{B}_n\left(R_{k_a}\right)\setminus \mathcal{B}_{n-1}\left(R_{k_a}\right)\right)\nonumber\\
& \geq & \pi R_{k_a}^2\left(\left( 1-\alpha\right) + \sum_{n=1}^{k_a-1}\left( 1-\frac{2}{\pi}\left(\frac{\pi}{2}-\frac{R_n}{2 R_{k_a}}\right)+\frac{R_n}{2\pi R_{k_a}}\sqrt{4-\frac{R_n^2}{R_{k_a}^2}}\right)\right)\nonumber\\
& \geq & \pi R_{k_a}^2\left(\left( 1-\alpha\right) + \sum_{n=1}^{k_a-1}\frac{R_n}{\pi R_{k_a}}\right)\nonumber\\
& \stackrel{(a)}{=} & V_{k_a}\left(\left( 1-\alpha\right) + \sum_{n=1}^{k_a-1}\frac{\sqrt{V_n}}{\pi \sqrt{V_{k_a}}}\right)\nonumber\\
& \stackrel{(b)}{\geq} & V_{k_a}\left( 1-\alpha\right) + \sum_{n=1}^{k_a-1}\frac{V_{n}}{\pi} = V\left(\mathcal{A}^{UB_2}\left(R_{k_a}\right)\right) .\nonumber
\end{eqnarray}
The equality (a) comes by change of variables and the inequality (b) results from the fact that $ V_{k_a}\geq V_n$. After bounding the area, we use Lemma \ref{LemmaAuxK} which states that, $R\acute{}_{k_a}\preceq_{st} R_{k_a}$. As a result, the area $\mathcal{A}^{UB_2}$ with radius $R\acute{}_{k_a}$ is stochastically smaller than the same area with radius $R_{k_a}$. The upper bound is further calculated using the p.d.f. in (\ref{pdfANC})

\begin{eqnarray}
\label{LB}
k_a!\ \int_0^{\infty}dV\acute{}_1e^{V\acute{}_1}\cdots\int_{V\acute{}_{k_a-2}}^{\infty}dV\acute{}_{k_a-1}e^{V\acute{}_{k_a-1}}\int_{V\acute{}_{k_a-1}}^{\infty}dV\acute{}_{k_a} e^{-k_aV\acute{}_{k_a}}e^{-\left(1-\alpha\right)V\acute{}_{k_a}} e^{-\sum_{n=1}^{k_a-1}\frac{V\acute{}_{n}}{\pi}} =\nonumber\\
k_a!\ \int_0^{\infty}dV\acute{}_1e^{\left(1-\frac{1}{\pi}\right)V\acute{}_1}\cdots\int_{V\acute{}_{k_a-2}}^{\infty}dV\acute{}_{k_a-1}e^{\left(1-\frac{1}{\pi}\right)V\acute{}_{k_a-1}}\int_{V\acute{}_{k_a-1}}^{\infty}dV\acute{}_{k_a} e^{-k_aV\acute{}_{k_a}}e^{-\left(1-\alpha\right)V\acute{}_{k_a}}.\nonumber
\end{eqnarray} 
\end{proof}

As in the case of the generation number, the bounds for the ancestor number probability are density-invariant as well (see \textbf{P.I}). We next give a numerical evaluation of the bounds in Table \ref{TabakaULB} and plot the results for comparison in Fig. \ref{FigAkaplotA}. 

\begin{table}[h]
\caption{Numerical values for the bounds of the ancestor number probability $g_a(k_a)$, for $k_a=1,\ldots,6$.}
\centering
\begin{tabular}{|l|l|l|l|l|l|l|l|l|l|l|}
\hline
			& $g_a(1)$ & $g_a(2)$	& $g_a(3)$	& $g_a(4)$	& $g_a(5)$ & $g_a(6)$\\\hline                
Upper Bound & $0.6215$ & $0.3977$ & $0.2529$ & $0.1593$ & $0.0996$ & $0.0618$\\      
Lower Bound & $0.6215$ & $0.1238$ & $0$ & $0$ & $0$ & $0$\\              
\hline
\end{tabular}
\label{TabakaULB}
\end{table}

\begin{figure}[h]
\centering
 \includegraphics[trim = 15mm 60mm 00mm 55mm, clip, width=0.6\textwidth]{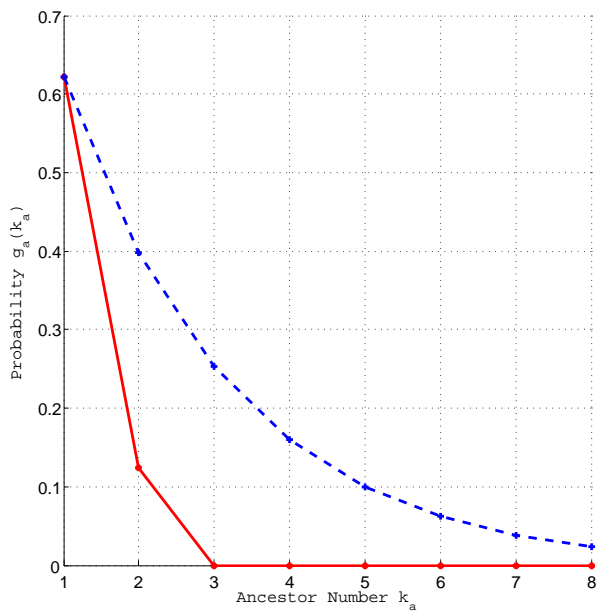}
	\caption{Evaluation of the upper and lower bounds of the ancestor number probability $g_a(k_a)$, for $k_a\geq 1$.}
\label{FigAkaplotA}
\end{figure}

\subsection{Interpretation of the ancestor number: Cluster size distribution}

The ancestor number is, as explained at the beginning of the previous paragraph, a more precise number to describe the cluster size, compared to the generation number.
A first important observation is that the set of events $\mathcal{E}_{k_a}=$ \{the typical atom has ancestor number $k_a$\} exhausts the event space and hence,
\begin{eqnarray}
\label{AncNumPDF}
\sum_{k_a=1}^{\infty}g_a\left(k_a\right) & = & 1.
\end{eqnarray}
To get a better understanding of the ancestor number, consider a finite set of atoms, with size $N$. The number of atoms will tend to infinity $N\rightarrow \infty$ as we approach the actual realization of the p.p.p. $\Phi$. Let the algorithm described in the previous paragraph evolve stepwise, and first find all atoms with ancestor number $1$ (in step 1). These are $N_1$ out of the total $N$ atoms and constitute the "roots" of all clusters with higher size. In other words, the total number of existing clusters in the finite setting is $N_1/2$ (because both mutually-nearest-neighbours take ancestor number $1$). A subset of these will add (in Step 2) another atom with ancestor number $2$. The number of clusters with cardinality at least $3$ will be $N_2\leq N_1/2$. The process will continue in the next steps of the algorithm, resulting in the inequalities
\begin{eqnarray}
\label{NumOFclust}
\frac{N_1}{2}\geq N_2\geq N_3 \geq \ldots
\end{eqnarray}
Dividing the above expression by $N$ and letting $N\rightarrow\infty$ we have proven that:
\begin{Lem}
\label{LemIneqAN}
The probability that the ancestor number of the typical atom is $k_a$, is ordered as follows
\begin{eqnarray}
\label{NumOFclust}
\frac{g_a\left(1\right)}{2}\geq g_a\left(2\right)\geq g_a\left(3\right) \geq \ldots
\end{eqnarray}
\end{Lem}
This Lemma further tightens the upper bound of the ancestor number probability, as shown in Fig.\ref{06215div2}.

\begin{figure}[h]
\centering
 \includegraphics[trim = 15mm 65mm 00mm 55mm, clip, width=0.6\textwidth]{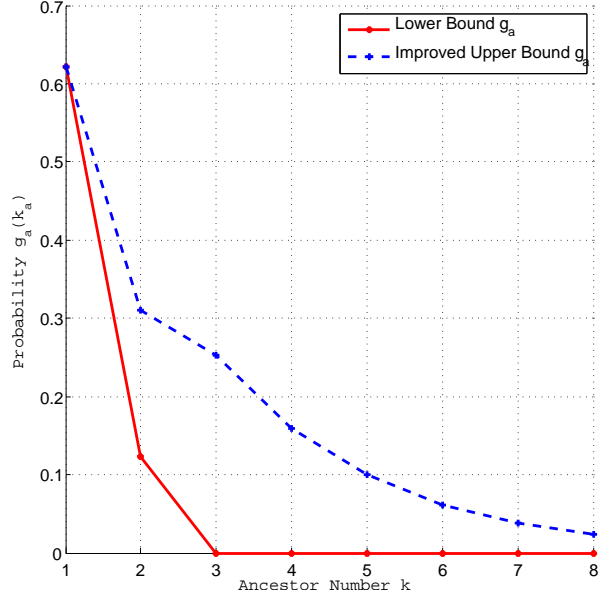}
	\caption{Tightened upper bound of the ancestor number probability $g_a(k_a)$, for $k_a\geq 1$.}
\label{06215div2}
\end{figure}
Using the above line of thought, the probability that the cardinality of the typical cluster is larger or equal to $k_a+1$ is exactly defined by the probability of the ancestor number $k_a$.

\begin{Thm}
\label{ThmClusterSize}
The probability that the typical atom belongs to a cluster of size $k_a+1$, is fully defined by the probability of the typical atom having ancestor number $k_a$, i.e.
\begin{eqnarray}
\label{CardClustProb}
\mathbb{P}\left[\mathbf{z}_o\in\mathcal{C}_o,\ card\left(\mathcal{C}_o\right) = k_a+1\right] & = & \left(k_a+1\right)\left(g_a(k_a) - g_a\left(k_a+1\right)\right),\ k_a\geq 2\nonumber\\
\label{CardClustProb2}
\mathbb{P}\left[\mathbf{z}_o\in\mathcal{C}_o,\ card\left(\mathcal{C}_o\right) = 2\right] & = & 2\left(g_a(1)/2 - g_a\left(2\right)\right),\ k_a= 1.
\end{eqnarray}
The tail probability of the above event is equal to
\begin{eqnarray}
\label{CardClustProbTail}
\mathbb{P}\left[\mathbf{z}_o\in\mathcal{C}_o,\ card\left(\mathcal{C}_o\right)\geq k_a+1\right] & = & k_a g_a\left(k_a\right) + \sum_{n=k_a}^{\infty} g_a\left(n\right),\ k_a\geq 2
\end{eqnarray}
From (\ref{CardClustProb2}) and (\ref{AncNumPDF}), $\mathbb{P}\left[card\left(\mathcal{C}_o\right)\geq 2\right] = \sum_{n=1}^{\infty} g_a\left(n\right) = 1$.
\end{Thm}

\begin{proof}
To prove this, we use the finite model above with $N$ total atoms and $N_{k_a}$ atoms with ancestor number $k_a$ (see also (\ref{NumOFclust})). Suppose the maximum cluster size is $K> 2$. Then there are $N_{K-1}$ atoms with ancestor number $k_a=K-1$. Similarly, there are $N_{K-2},\ldots,N_1$ atoms with ancestor number $k_a=K-2,\ldots,1$. Since, the size of the maximum cluster is $K$, the total atoms which belong to these clusters is $KN_K$. Then,
\begin{eqnarray}
\label{proofEqnArr}
\frac{KN_{K-1}}{N}\rightarrow \mathbb{P}\left[\mathbf{z}_o\in\mathcal{C}_o,\ card\left(\mathcal{C}_o\right)= K\right].\nonumber \\
\frac{(K-1)\left(N_{K-2}-N_{K-1}\right)}{N}\rightarrow \mathbb{P}\left[\mathbf{z}_o\in\mathcal{C}_o,\ card\left(\mathcal{C}_o\right)= K-1\right].\nonumber \\
\ldots\nonumber\\
\frac{3\left(N_{2} - N_3\right)}{N}\rightarrow \mathbb{P}\left[\mathbf{z}_o\in\mathcal{C}_o,\ card\left(\mathcal{C}_o\right)= 2\right].\nonumber
\end{eqnarray}
For the case $k_a=1$, we have a small difference, because both atoms of the root take ancestor number equal to $1$. Hence
\begin{eqnarray}
\label{proofEqnArr2}
\frac{2\left(N_{1}/2 - N_2\right)}{N}\rightarrow \mathbb{P}\left[\mathbf{z}_o\in\mathcal{C}_o,\ card\left(\mathcal{C}_o\right)= 2\right].\nonumber
\end{eqnarray}
By letting $N\rightarrow\infty$ the left-hand side gives the probability that the typical atom belongs to a cluster of size $k_a+1$,
\begin{eqnarray}
\label{ExactProbANC}
\mathbb{P}\left[\mathbf{z}_o\in\mathcal{C}_o,\ card\left(\mathcal{C}_o\right) = k_a+1\right] & = & \left(k_a+1\right)\left(g_a(k_a) - g_a\left(k_a+1\right)\right),\ k_a\geq 2\nonumber\\
\mathbb{P}\left[\mathbf{z}_o\in\mathcal{C}_o,\ card\left(\mathcal{C}_o\right) = 2\right] & = & 2\left(g_a(1)/2 - g_a\left(2\right)\right),\ k_a= 1.
\end{eqnarray}
\end{proof}
The tail distribution of the cluster size, is bounded using the bounds of the ancestor number in Theorem \ref{AncestorBoundsUD} and Lemma \ref{LemIneqAN}. It is plotted in Fig.\ref{BoundTailClust}. Furthermore, using the known bounds for the ancestor number (upper bound $g_a\leq g_a^{UB}$ for the positive term of the difference and lower bound $-g_a\leq -g_a^{LB}$ for the negative term, we can find a loose upper bound for the exact cluster size. The most important result from the plots is to illustrate how the probability of a higher cluster size (Fig.\ref{BoundClustSize}) diminishes as $k_a+1$ increases. Observe that the loose upper bound is already $10\%$ for cluster size $11$, implying that cluster sizes of average or large size are highly improbable. This is better supported by the tail probability upper bound, which is also around $10\%$ for cluster size $11$.

\begin{figure}[h]
\centering
 \includegraphics[trim = 15mm 65mm 00mm 55mm, clip, width=0.6\textwidth]{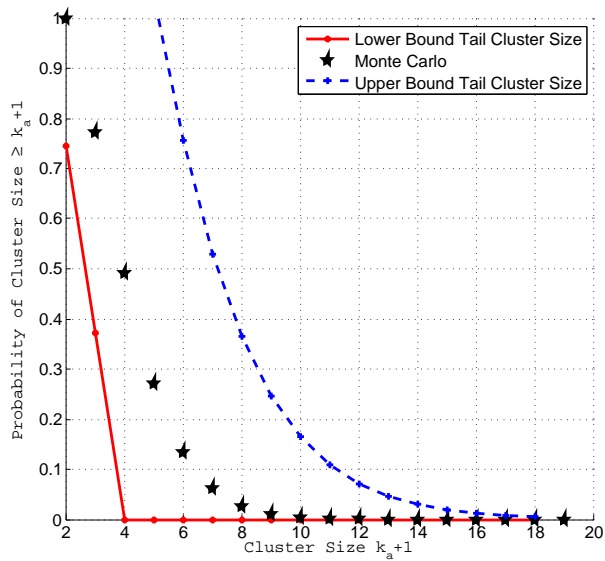}
	\caption{Distribution of the tail probability for the cluster size $\mathbb{P}\left[card\left(\mathcal{C}_o\right)\geq k_a+1\right]$.}
\label{BoundTailClust}
\end{figure}

\begin{figure}[h]
\centering
 \includegraphics[trim = 15mm 65mm 00mm 55mm, clip, width=0.6\textwidth]{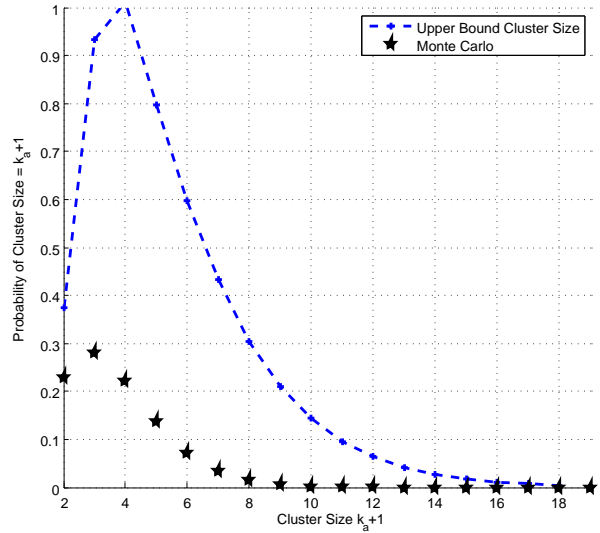}
	\caption{Upper bound for the probability of the cluster size $\mathbb{P}\left[card\left(\mathcal{C}_o\right)= k_a+1\right]$, $k_a\geq 2$.}
\label{BoundClustSize}
\end{figure}

The two figures also include the Monte Carlo values of the cluster size probability. These are derived by simulation over an area with dimensions $10\times 10$ [$m^2$] and a p.p.p. with density $2$ [atoms/$m^2$] hence in average $\mathbb{E}\left[N_t\right]=200$ atoms are randomly placed per realization. The values of the cluster size probability result after $5000$ iterations.

\begin{table}[h]
\caption{Numerical values for the bounds and Monte Carlo simulation of the cluster size and tail probability, for $k=2,\ldots,8$.}
\centering
\begin{tabular}{|l|l|l|l|l|l|l|l|}
\hline
$card(\mathcal{C}_o)=$			& $2$ & $3$	& $4$	& $5$	& $6$ & $7$ & $8$\\\hline
Upper Bound & $0.3739$ &  $0.9323$ & $1.0116$ & $0.7967$ & $0.5974$ & $0.4325$ & $0.3050$\\\hline
Monte Carlo & $0.2287$ &  $0.2796$ & $0.2219$ & $0.1362$ & $0.0727$ & $0.0347$ & $0.0163$\\
\hline
\hline
$card(\mathcal{C}_o)\geq$			& $2$ & $3$	& $4$	& $5$	& $6$ & $7$ & $8$\\\hline
Upper Bound & $1.6037$ &  $1.6037$ & $1.4301$ & $1.0558$ & $0.7570$ & $0.5303$ & $0.3647$ \\\hline
Monte Carlo & $0.9999$ &  $0.7712$ & $0.4916$ & $0.2697$ & $0.1335$ & $0.0608$ & $0.0261$ \\\hline
Lower Bound & $0.7453$ &  $0.3714$ & $0$ & $0$ & $0$ & $0$ & $0$\\
\hline
\end{tabular}
\label{TabgkULB}
\end{table}

%\section{Clusters in the lilypond (LL) graph}

\newpage

\bibliographystyle{alpha}
\footnotesize
\bibliography{MatchingCOOP}

\end{document}